\documentclass[12pt]{amsart}
\usepackage{enumitem}
\usepackage{float}
\usepackage[hmargin=0.8in,height=8.8in]{geometry}
\usepackage[pdftex]{graphicx}
\usepackage{amssymb,amsthm, times, mathptmx}
\usepackage{delarray,verbatim}
\usepackage{soul}
\usepackage[font=scriptsize]{caption}
\usepackage[font=scriptsize]{subcaption}
\usepackage{srcltx}
\usepackage{ifpdf}
\ifpdf
\usepackage{bigints}
 \else
\usepackage[dvips]{graphicx}
 \fi
\graphicspath{{./dissertation-images/}}
\usepackage[font=scriptsize]{caption}

\usepackage[numbers]{natbib}
\usepackage{ifpdf}
\usepackage{color}
\definecolor{webgreen}{rgb}{0,.5,0}
\definecolor{webbrown}{rgb}{.8,0,0}
\definecolor{emphcolor}{rgb}{0.95,0.95,0.95}

\usepackage{hyperref}
\hypersetup{%
	colorlinks=true,
	linkcolor=webbrown,
	filecolor=webbrown,
	citecolor=webgreen,
	breaklinks=true}
\numberwithin{equation}{section}

\newtheorem{theorem}{Theorem}[section]
\newtheorem{proposition}{Proposition}[section]
\newtheorem{corollary}{Corollary}[section]
\newtheorem{remark}{Remark}[section]
\newtheorem{lemma}{Lemma}[section]

\numberwithin{remark}{section} \numberwithin{proposition}{section}
\numberwithin{corollary}{section}

\newcommand {\R}{\mathbb{R}}

\newcommand {\F}{\mathcal{F}}

\newcommand {\p}{\mathbb{P}}
\newcommand {\G}{\mathfrak{G}}
\newcommand {\E}{\mathbb{E}}

\newcommand{\diff}{{\rm d}}

\newcommand{\word}{\hspace{0.2cm}}
\newcommand{\conn}{\quad\text{and}\quad}
\newcommand{\1}{\mbox{1}\hspace{-0.25em}\mbox{l}}

\newcommand{\II}{\mathcal{I}}

\newcommand{\ii}{\int_0^\infty}

\title[Time Reversal and Last Passage Time of Diffusions]{Time Reversal and Last Passage Time of Diffusions with Applications to Credit Risk Management}
\author[Egami]{Masahiko Egami}
\address[M. Egami]{Graduate School of Economics,
	Kyoto University, Sakyo-Ku, Kyoto, 606-8501, Japan}
\email{egami@econ.kyoto-u.ac.jp}
\urladdr{http://www.econ.kyoto-u.ac.jp/{\textasciitilde}egami/}
\thanks{The first author is supported by Grant-in-Aid for Scientific Research (C) No. 18K01683, Japan Society for the Promotion of Science. The second author is in part supported
	by JSPS KAKENHI Grant Number JP 17J06948.}

\author[Kevkhishvili]{RUSUDAN KEVKHISHVILI}
\address[R. Kevkhishvili ]{Graduate School of Economics,
	Kyoto University, Sakyo-Ku, Kyoto, 606-8501, Japan}
\address{Research Fellow of Japan Society for the Promotion of Science}
\email{keheisshuiri.rusudan.73m@st.kyoto-u.ac.jp}

\date{}
\begin{document}

\begin{abstract}
We study time reversal, last passage time, and $h$-transform of linear diffusions. For general diffusions with killing, we obtain the probability density of the last passage time to an arbitrary level and analyze the distribution of the time left until killing after the last passage time. With these tools, we develop a new risk management framework for companies based on the leverage process (the ratio of a company asset process over its debt) and its corresponding alarming level. We also suggest how a company can determine the alarming level for the leverage process by constructing a relevant optimization problem.
\end{abstract}
\maketitle \noindent \small{\textbf{Key Words:} Time reversal; linear diffusion; last passage time; $h$-transform; risk management.
	%\newline \noindent JEL Classification: G32; G33
 }

\section{Introduction}\label{sec:introduction}
In this paper, we study last passage time to a specific state and time reversal of linear diffusions and consider their applications to credit risk management. Specifically, we deal with a \emph{general time-homogeneous transient linear diffusion} process $X$ that has a killing boundary. We address three problems concerning diffusion $X$. First, we fix an arbitrary level, which we denote by $\alpha$ and call ``reference point" in the mathematical context or ``alarming (or premonition) level" in the financial context. This is an entrance point to a certain region which we refer to as dangerous zone. We study the distribution of the last passage time to this point. Second, we derive the distribution of the time between the last passage time to $\alpha$ and killing time. Finally, we suggest how the level $\alpha$ of interest can be chosen as a solution to an optimization problem in the context of credit risk management.
\newline \indent Last passage times of standard Markov processes are studied in \citet{getoor-sharpe} where they analyze the joint distribution of the last exit time of the process from a transient set and its location at that time. \citet{pitman-yor} study the density of last exit time of regular linear diffusions on positive axis with the scale function satisfying $s(0+)=-\infty$ and $s(\infty)<\infty$ by using Tanaka's formula and apply the result to Bessel processes. We study last passage times in a general setting and employ the technique presented in \citet{salminen1984}, which uses the $h$-transform method. This technique is a fast and easy way to obtain an explicit formula. In Proposition \ref{prop:1} we treat various cases comprehensively, providing a method with which one can find the distribution of last passage time of general diffusions, irrespective of the spacial relationship of the starting point and reference point, or whether killing time is almost surely finite or not.  The density of the last passage time to level $\alpha$ can be used for a premonition of imminent killing: how dangerous would it be if the process hits that level $\alpha$?  We apply our mathematical results to the leverage process of the company which is a function of a regular time-homogeneous linear diffusion in our setting (see Section \ref{subsec:2-2} for the financial model). We discuss this application in Section \ref{sec:intro-application} in detail.
\newline\indent Section \ref{subsec:4-general} is concerned with the time left until killing after the last passage time to $\alpha$ which is the second problem of our interest. While the literature mostly deals with the total time spent in a dangerous zone, that is, occupation time, we study the time after the last passage time to level $\alpha$. This is essential information for risk management and is one of the novel features of this section. Since this problem involves two random variables (last passage time and time of death), it is complex; therefore, we use the reversed process of diffusions to make the problem simpler. See \citet{sharpe1980} and \citet{williams1974} for a specific example of time reversal and a recent account by \citet{chung-walsh}. Proposition \ref{prop:reversal} derives the distribution of the time left until killing (after the last passage time to $\alpha$) for a general diffusion $X$. In case of a Brownian motion with drift, we obtain semi-explicit expression of the density of the time until killing in Proposition \ref{prop:additional}. We cannot find any articles that handle this problem to the best of our knowledge.
\newline \indent Finally, in Section \ref{sec:endogenizing}, we suggest a method to choose an appropriate level $\alpha^*$ for our general diffusion $X$. For this, we formulate an optimization problem, which we believe is new, using the last passage time arguments and occupation time distribution. This problem is constructed in the context of credit risk management but can be applied to various problems. We continue this section with discussing the application of our theoretical results to risk management.
\subsection{Application to Credit Risk Management Framework}\label{sec:intro-application}
\indent We apply general results for a diffusion $X$ (Propositions \ref{prop:1} and \ref{prop:reversal}) to the leverage process of the company. In our financial model (Section \ref{subsec:2-2}), the leverage process is a function of a regular time-homogeneous linear diffusion. We are interested in a certain threshold, denoted by $R^*$, of a company's leverage ratio, an exit from which means an entry into a dangerous zone and leads to insolvency without returning to $R^*$. In our setting, the level $R^*$ for the leverage ratio is equivalent to a certain level $\alpha$ for the underlying diffusion; therefore, we continue the discussion using $\alpha$.

\citet{elliot2000} discuss valuation of defaultable claims when the payoff depends on last passage time of a firm's value to a certain state (see also \citet{jeanblanc_rutkowski}). We study the last passage time to $\alpha$ (denoted by $\lambda_{\alpha}$) before insolvency for the leverage process and analyze its implications regarding credit risk. It is well-known that last passage time is not a stopping time. We take this point into consideration by using the optional projection (see Remark \ref{rmrk:projection}). The \emph{last} passage to $\alpha$ indicates that the company cannot recover to normal business conditions once this occurs. It is often the case that companies in financial distress cannot recover once the leverage ratio deteriorates to a certain level: the lack of creditworthiness makes it almost impossible to continue usual business relations with their contractors, suppliers, customers, creditors, and investors, which further pushes the company to the brim of insolvency. In this sense, $\alpha$ can be considered as a precautionary level, the passage of which triggers an alarm. \emph{While $\lambda_{\alpha}$ is not a stopping time, the density obtained in Corollary \ref{coro:1} enables us to calculate certain probabilities associated with $\lambda_{\alpha}$.} We emphasize that all of these probabilities are calculated \emph{based on the current position of the leverage process and the information available up to current time.} This is a novel approach to analyzing the dynamics of the leverage ratio and there are no other studies to the best of our knowledge.
%Based on this density $\p_\cdot(\lambda_\alpha\in \diff t)$, we can compute different probabilities associated with $\lambda_{\alpha}$. As an example, we can compute $\int_0^s\p_\cdot(\lambda_\alpha\in \diff t)$, the probability that the premonition occurs within $s$ period of time.
As we demonstrate by using actual company data in Section \ref{subsec:last-passage-application}, the probabilities associated with $\lambda_{\alpha}$ provide very useful information for risk management.  

Together with the last passage time, we study the time left until insolvency after leverage process passes $\alpha$ for the last time. There are no other studies that analyze this time interval. The expectation of this random variable can be computed at time zero based on the drift and volatility parameters of the leverage process. We obtained the density of the time until insolvency after $\lambda_{\alpha}$ semi-explicitly in Proposition \ref{prop:additional}. See Section \ref{subsec:4-special} and Figure \ref{Fig-problem2}, where we used the actual company data for the analysis. 
\newline \indent Finally, the appropriate level of $\alpha^*$ for the leverage ratio, passing of which should trigger an alarm to the management, can be determined by solving the optimization problem in Section \ref{sec:endogenizing}. This problem involves last passage time and occupation time of a dangerous zone. Below we cite some articles that are related to our optimization problem. \citet{gerber} model the surplus of the company by using Brownian motion with drift and use the Omega model to analyze occupation time in ``red''. This model assumes that there is a time interval between the first instance of the company's surplus becoming negative and bankruptcy. They study the total time Brownian motion with positive drift spends below zero and the relation between the Laplace transform of this occupation time and the probability of going bankrupt in finite time. \citet{albrecher} model the surplus as a spectrally-negative Markov additive process and assume the surplus is observed at arrival times of an independent Poisson process. They assume that the rate of observations depends on the current state of the environment. In the same vein, we change the parameters of the asset process when the leverage ratio is below a certain threshold in Section \ref{subsubsec:alpha-effect}, where we analyze the effect of different managerial strategies.

This article lays out the  mathematical foundation of the last passage time that is vital for practical implementation. By using the company's balance sheet data and equity time series for estimating drift and variance parameters of its asset value process, through the method described in Section \ref{subsubsec:method}, the management can determine the threshold level $\alpha^*$ below which the company should be operated on alert and with precaution. For this fixed $\alpha^*$, the management can compute different probabilities associated with the last passage time $\lambda_{\alpha^*}$ and the distribution of time between $\lambda_{\alpha^*}$ and the insolvency. We believe that this paper can contribute to more refined credit risk management for companies, as illustrated in Section \ref{sec:application} with empirical analysis and summarized in Section \ref{sec:conclusion}. The data for all the estimation in the paper was obtained from Thomson Reuters Datastream.
%Numerical integration in Tables \ref{tbl:american-apparel-2} and \ref{tbl:american-apparel-alpha} was done in Mathematica 10.1. Figures \ref{Fig-problem2}, \ref{fig:alpha_gamma}, and \ref{Fig-problem2-2} and the values in Table \ref{tbl:comparative_alpha} were produced in Matlab 2015a.

\section{Mathematical Framework}
Let us consider the probability space $(\Omega, \F, \p)$, where $\Omega$ is the set of all possible realizations of the
stochastic economy, and $\p$ is a probability measure defined on $\F$. We denote by
$\mathbb{F}=\{\F_t\}_{t\ge 0}$ the filtration satisfying the usual conditions and consider a regular time-homogeneous diffusion process $X$ adapted to $\mathbb{F}$. The state space of $X$ is $\mathcal{I}=(\ell, r)\in \R$ and we adjoin an isolated point $\Delta$ to $\mathcal{I}$. We call $\omega$ a sample path from $[0, \infty)\mapsto \mathcal{I}\cup \Delta$ with coordinates $\omega_t=X_t(\omega)$. The lifetime of $X$ is defined by $\zeta(\omega):=\inf\{t\ge 0: X_t(\omega)\notin \II\}$. Let $\p_y$ denote the probability measure associated to $X$ when started at $y\in \mathcal{I}$. Let $(P_t)_{t\ge 0}$ be the transition semigroup of $X$. For $X$ with infinitesimal parameters $\mu(x)$ and $\sigma^2(x)$, the generator is defined by
\[\G f(x)=\mu(x)f'(x)+\frac{1}{2}\sigma^2(x)f''(x)\]
for a twice continuously differentiable function $f(x)$ on $\mathcal{I}$.

Let the first passage time be denoted by $T_x=\inf\{t\ge 0: X_t=x\}$ for $x\in \mathcal{I}$.  For $\ell<a\le y\le b<r$, the scale function $s(\cdot)$ of $X$ satisfies
\[
\p_y(T_a<T_b)=\frac{s(b)-s(y)}{s(b)-s(a)}.
\] For more information about diffusion processes, see, e.g., \citet[chapter II]{borodina-salminen}).
\subsection{$h$-transform}
Let $h$ be an excessive function; that is, $h$ is a nonnegative Borel measurable function with the properties
\begin{align*}
&h\ge P_t h, \quad \forall t\ge 0 \conn h=\lim_{t\downarrow 0}P_t h.
\end{align*}
% and let $\mathbf{E}_h:=\{x: 0<h(x)<+\infty\}$.
For a Borel measurable set $A\in \mathcal{B}(\mathcal{I})$, define for $u, v\in \mathcal{I}$
\begin{equation}\label{eq:h-transform}
P^h(t; u, A):=\frac{1}{h(u)}\int_A h(v) P(t; u, \diff v).
\end{equation} The $h$-transform of $X$ is a regular diffusion with the transition function \eqref{eq:h-transform}. Following Salminen \cite{salminen1985}, let us call an excessive function $h$ \emph{minimal} if the $h$-transform of $X$ converges $\p^h_{.}$-a.s. to a single point, that is, for all $y\in \mathcal{I}$
\[
\p_y^h(\lim_{t\rightarrow \zeta}X_t=z)=1
\] for some $z\in [\ell, r]$.  Note that $\p_y^h$ is the probability law of the $h$-transform of $X$ starting at $y$.
\subsection{Leverage Process}\label{subsec:2-2}
We use a \emph{structural approach} proposed by \citet{merton1974} and analyze the credit-worthiness of a company through the behavior of its unobservable asset process (firm's value). This approach models the firm's value as a geometric Brownian motion and assumes the company equity is a European call option written on the asset process with a strike price equal to the value of debt at maturity. The structural approach is widely implemented in practice and one of the examples is Expected Default Frequency (EDF) model provided by Moody's Analytics. For more information about structural models, we refer the reader to Section 10.3 in \citet{mcneil-frey2015}.
\newline \indent Suppose that a firm has total assets with market value $A=(A_t)_{t\ge 0}$. We assume that the asset process $A$ follows geometric Brownian motion with parameters $\nu\in \R$ and $\sigma>0$, and the debt process $D=(D_t)_{t\ge 0}$ grows at the risk-free rate of $r$:
\[
\diff A_t =\left(\nu+\frac{1}{2}\sigma^2\right) A_t \diff t + \sigma A_t \diff B_t^A\conn
\diff D_t = rD_t \diff t
\] where we set initial values $A_0$ and $D_0$, respectively. $B^A=\left(B_t^A\right)_{t\geq0}$ denotes a standard Brownian motion adapted to $\F_t$. By assuming $A_0>D_0$, we define the leverage process $R=(R_t)_{t\ge 0}$ as
$R_t := \frac{A_t}{D_t}$. Then, we set the insolvency time of the firm as \[
T:=\inf\{t\ge 0: R_t =1\}.
\]
Since $A_t=A_0e^{\nu t+\sigma B_t^A}$ and $D_t=D_0e^{rt}$,
\[R_t=\frac{A_0}{D_0}\exp\left(\left(\nu-r\right)t + \sigma B_t^A\right)\]
and $R_T=1$ implies
\[\frac{\nu-r}{\sigma}T+B_T^A=\frac{1}{\sigma}\ln \left(\frac{D_0}{A_0}\right):=c\]
which means that the insolvency time is the first passage time of Brownian motion with drift $\frac{\nu-r}{\sigma}$ and unit variance parameter to state $c$.
Consequently, our study about the leverage process $R_t$ can be reduced to the study of the Brownian motion with drift and unit variance parameter:
\begin{equation}\label{eq:dynamics-of-X}
X_t:=y+ \frac{\nu-r}{\sigma} t + B_t^A, \quad t\ge 0
\end{equation}
on the state space $\mathcal{I}=(c, \infty)$. It follows that $y=0$ in our model; however, we continue the discussion for an arbitrary $y$ for the purpose of making general statements. We have
\begin{equation*}
T=T_c:=\inf\{t\ge 0: y+\mu t + B_t^A=c\} \quad \text{with \word $\mu:=\frac{\nu-r}{\sigma}$}.
\end{equation*}

Since the stopping time $T$ is predictable, it is possible and may be a good idea to set a threshold level $R^*$ for the leverage process, so  that when it passes this point from above, the firm should prepare and start precautionary measures to avoid possible subsequent insolvency. $R_t=R^*$ means that $X_t=\frac{1}{\sigma}\ln \left(\frac{R^* D_0}{A_0}\right):=\alpha$, and we can again study the passage time to this arbitrary $\alpha\in (c, r)$ for the Brownian motion with drift starting from $y$.

We discuss and prove our results for a generic diffusion $X$. We wish to stress that our assertions below hold in a general setting and thus are applicable to other problems: the general results are Propositions \ref{prop:1} and \ref{prop:reversal}. Hence, the results for our leverage process are derived as their respective Corollaries \ref{coro:1} and \ref{coro:2}. In addition, we obtain a representation of the density for the time until insolvency for the leverage process in Proposition \ref{prop:additional}.  While we use the same $X$ for general diffusion and the specific leverage process, we have made sure that the reader would not be confused.

\section{Mathematical Results} \label{sec:general}
Let $X=\{\omega_t, t\ge 0, \p_y\}$ on $\mathcal{I}=(\ell, r)$ with $s(\ell)>-\infty$ and $s(r)=+\infty$ (or $s(r)<+\infty$) be a regular time-homogeneous transient diffusion process with lifetime $\zeta=\inf\{t: \: X_t\notin \mathcal{I}\}$. The left boundary is regular with killing and the right boundary is natural. Henceforth, we refer to the case of $s(\ell)>-\infty$ and $s(r)=+\infty$ as Case 1 and refer to the case of $s(\ell)>-\infty$ and $s(r)<+\infty$ as Case 2.
\subsection{The Last Passage Time}\label{subsec:last-passage}
Let $\lambda_x$ denote the last passage time to the level $x$. Then, $\lambda_x=\sup\{t\geq0: \omega_t=x\}$ with the convention $\sup\{\emptyset\}=0$. Before deriving the distribution of the last passage time, let us introduce some objects that are needed.\\

\noindent \underline{Case 1:}
In Case 1, the functions
\begin{align*}
k_{\ell}(x)&=1 \conn k_r(x)=s(x)-s(\ell), \quad x\in \mathcal{I},
\end{align*} are minimal excessive. The excessiveness is because they are harmonic (see Theorem 12.4 of Dynkin \cite{dynkin}). Now, $k_{\ell}(x)=1$ is minimal because the (original) diffusion converges to the boundary point $\ell$ (since $s(\ell)>-\infty$ and $s(r)=+\infty$), that is, $\p_y(\lim_{t\rightarrow \zeta}X_t=\ell)=1$ for all $y\in \mathcal{I}$. Similarly, if we consider the $k_r$-transform of $X$, it is a regular diffusion with a transition function as in \eqref{eq:h-transform}
\[P^{k_r}(t; u, A)=\frac{1}{k_r(u)}\int_A k_r(v) P(t; u, \diff v)\]
for a Borel measurable set $A$. Its scale function is $s^{r}(x)=-\dfrac{1}{s(x)-s(\ell)}$. Then, we can see that for all $y\in \mathcal{I}$, we have $\p_y^{k_r}(\lim_{t\rightarrow \zeta} X_t=r)=1$ since $s^r(\ell)=-\infty$.

Note that the Green function (see \citet{borodina-salminen} for the definition) for $X$ is
\begin{equation}\label{eq:G}
G(y, z)=\ii p(t; y, z)\diff t=(s(y)-s(\ell))\wedge (s(z)-s(\ell)), \quad y, z\in\mathcal{I},
\end{equation} where $p(t; y, z)$ is the transition density with respect to the speed measure. Let us consider
\begin{equation}\label{eq:kx-case1}
k_x(y):=\min\left(k_{\ell}(y),\dfrac{k_r(y)}{s(x)-s(\ell)}\right)=\p_y(\text{hit}\; x \; \text{before}\; \ell)=\p_y(\lambda_x>0)=\begin{cases} \dfrac{s(y)-s(\ell)}{s(x)-s(\ell)}, \quad y\leq x,\\
\qquad 1, \quad \quad \quad \; y>x.
\end{cases}
\end{equation}
Note that $k_x$ is excessive since it is the minimum of the two excessive functions (see Proposition 3.2.2 in \citet{chung-walsh}).
As in \eqref{eq:h-transform}, the $k_x$-transform of $X$ is a regular diffusion with transition density function
\[P^{k_x}(t; u, A)=\frac{1}{k_x(u)}\int_A k_x(v) P(t; u, \diff v)\]
for a Borel set $A$.
It can be seen that
\begin{align}\label{eq:identity-pre}
&\text{The $k_x$-transform (or $k_x$-diffusion) is identical in law}\\
&\text{ with $X$ conditioned to hit $x$ and be killed at its last exit time from $x$.} \nonumber
\end{align}
Indeed, for any $u,v\in\mathcal{I}$, the density $p^*$ of such conditioned diffusion with respect to the speed measure satisfies
\begin{align*}
\nonumber
p^*(t;u,v)m^*(\diff v)&=\p_u(X_t\in \diff v, t<\lambda_x\mid \lambda_x>0)=\frac{\E_u[\1_{X_t\in \diff v}\1_{t<\lambda_x}\1_{\lambda_x>0}]}{\p_u(\lambda_x>0)}=\frac{\E_u[\1_{X_t\in \diff v}\1_{t<\lambda_x}]}{\p_u(\lambda_x>0)}\\
&=\frac{\E_u[\1_{X_t\in \diff v}\E_u[\1_{t<\lambda_x}\mid \mathcal{F}_t]]}{\p_u(\lambda_x>0)}=\frac{\E_u[\1_{X_t\in \diff v}\E_u[\1_{\lambda_x\circ \theta(t)>0}\mid \mathcal{F}_t]]}{\p_u(\lambda_x>0)}\\
&=\frac{\E_u[\1_{X_t\in \diff v}]\E_v[\1_{\lambda_x>0}]}{\p_u(\lambda_x>0)}=\frac{p(t;u,v)m(\diff v)k_x(v)}{k_x(u)},
\end{align*}
where $m$ denotes the speed measure of the original diffusion $X$ and $\theta(\cdot)$ is the shift operator. It follows that, for all $x\in \mathcal{I}$, $\p_y^{k_x}(\lim_{t\rightarrow \zeta} X_t=x)=1$, so that $k_x$ is minimal.\\

\noindent \underline{Case 2:} In Case 2, the minimal excessive functions are
\begin{align}\label{eq:minimal-case2}
k_{\ell}(x)=s(r)-s(x),\\
k_r(x)=s(x)-s(\ell),\nonumber
\end{align}
and
\begin{equation}\label{eq:kx-case2}
k_x(y)=
\begin{cases}\frac{s(y)-s(\ell)}{s(x)-s(\ell)}, \quad y\leq x,\\
\frac{s(r)-s(y)}{s(r)-s(x)}, \quad y>x,
\end{cases}
\end{equation}
with $x, y\in\mathcal{I}$ (e.g., see Theorem 2.10 in \citet{salminen1985}). The Green function for $X$ is
\begin{equation*}\label{eq:G1}
G(y, z)=\ii p(t; y, z)\diff t=
\begin{cases}
\frac{(s(y)-s(\ell))(s(r)-s(z))}{s(r)-s(\ell)}\quad \ell<y\leq z<r,\\
\frac{(s(z)-s(\ell))(s(r)-s(y))}{s(r)-s(\ell)}\quad \ell<z\leq y<r,\\
\end{cases}
\end{equation*} where $p(t; y, z)$ is the transition density with respect to the speed measure.

We have the following Proposition \ref{prop:1} concerning the distribution of $\lambda_x$ for a general $X$ that complements Proposition 4 and Corollary 6 in Salminen \cite{salminen1984}. Indeed, in the proof we provide detailed techniques applicable to various cases, which depend upon whether $s(r)$ is finite or infinite, and whether $\p_y(T_x<\infty)$ is equal to or less than $1$. 

\begin{proposition}\label{prop:1}
	Let $X$ be a diffusion process with state space $\mathcal{I}=(\ell, r)$, $\ell$ being a regular killing boundary. Suppose that (Case 1) $s(\ell)>-\infty$ and $s(r)=+\infty$, or (Case 2) $s(\ell)>-\infty$ and $s(r)<+\infty$. Then, for any $y, x\in \mathcal{I}$, 
	%when $\lim\limits_{t\to\zeta}X_t=\ell$, 
	we have the following result:
	\begin{align}\label{eq:general-lambda}
	&\p_y(\lambda_x\in \diff t,\lambda_x>0)=\frac{p(t; y, x)}{s(x)-s(\ell)}\diff t \quad \textnormal{in Case 1},\\
	&\p_y(\lambda_x\in \diff t,\lambda_x>0,T_{\ell}<T_r)=\frac{p(t; y, x)}{s(x)-s(\ell)}\diff t \quad \textnormal{in Case 2}
	\end{align}
	where $p(t; y, x)$ is the transition density of $X$ with respect to its speed measure $m(\cdot)$.
\end{proposition}
\begin{remark}\label{rmrk:projection}
	\textnormal{Since the last passage time $\lambda_x$ is not a stopping time, we use optional projection (see \citet[Chpt 6. Sec.3]{rw-2}) to make it informative at arbitrary fixed time $t$. Specifically, we consider the optional projection of $\1_{\lambda_x>t}$ on the filtration $\mathbb{F}$ given by $\p_y(\lambda_x>t\mid \F_t)$ which is $\F_t$-measurable. We have
		\[\p_y(\lambda_x>t\mid \F_t)=\1_{t<T_{\ell}}\p_{X_t}(\lambda_x>0)=\begin{cases}
		\1_{t<T_{\ell}}\dfrac{s(X_t)-s(\ell)}{s(x)-s(\ell)} \quad \textnormal{if} \quad X_t\le x\\
		\1_{t<T_{\ell}}\dfrac{s(r)-s(X_t)}{s(r)-s(x)} \quad\textnormal{if} \quad X_t>x
		\end{cases}\]
		Note that since $s(r)=+\infty$ in Case 1, we would have $\p_y(\lambda_x>t\mid \F_t)=\1_{t<T_{\ell}}$ if $X_t>x$. We use both of the results, Proposition \ref{prop:1} and Remark \ref{rmrk:projection}, in the credit risk management application.} 
	%\noindent (2) Note that $\1_{\lambda_x>t}=\1_{\lambda_x\circ\theta_t>0}$ where $\theta$ is the shift operator. Therefore, using Markov property, 
	%\[\p_y(\lambda_x>t\mid \F_t)=\p_y(\lambda_x\circ\theta_t>0\mid \F_t)=\p_{X_t}(\lambda_x>0)=\p_{X_t}(\lambda_x>0,T_{\ell}<T_r)+\p_{X_t}(\lambda_x>0,T_r<T_{\ell}).\]
	%Using Proposition \ref{prop:1}, the first element satisfies
	%\[\p_{X_t}(\lambda_x>0,T_{\ell}<T_r)=\frac{\int_0^\infty p(s;X_t,x)\diff s}{s(x)-s(\ell)}.\]
	%Note that the second element is zero when we are dealing with diffusion in Case 1 ($s(r)=+\infty$). Using \eqref{eq:G}, we get the result for Case 1. Next we consider the second element in Case 2. We have 
	%\[\p_{X_t}(\lambda_x>0,T_r<T_{\ell})=\p_{X_t}(\lambda_x>0\mid T_r<T_{\ell})\p_{X_t}(T_r<T_{\ell}).\]
	%Therefore, we first consider diffusion $X$ conditioned on $T_r<T_{\ell}$. This is in fact a $k_r$-transform of the original diffusion $X$. Its characteristics are summarized in Remark \ref{rmrk-x-star} below. From Corollary 6 in \citet{salminen1984}, we obtain
	%\[\p_{X_t}(\lambda_x>0\mid T_r<T_{\ell})\p_{X_t}(T_r<T_{\ell})=\frac{\int_0^{\infty}p^{k_r}(s;X_t,x)\diff s}{s^r(r)-s^r(x)}\p_{X_t}(T_r<T_{\ell})=\frac{\int_0^{\infty}p(s;X_t,x)\diff s}{s(r)-s(x)}.\]
\end{remark}
\begin{proof}[Proof of Proposition \ref{prop:1}]
	First note that while we set $\ell$ as the killing boundary, it is for ease of exposition and our proof does not rely on this assumption.\\
	\underline{Case 1.}
	We have $s(\ell)>-\infty$ and $s(r)=\infty$. This includes the case when $\ell$ is a regular point with killing and $r$ is a natural boundary.
	%Let us recall some facts from the second paragraph of Section \ref{sec:last-passage-time}. The functions
	%\begin{align*}
	%k_{\ell}(x)&=1 \conn k_r(x)=s(x)-s(\ell), \quad x\in \mathcal{I}
	%\end{align*} are minimal excessive. The Green function for $X$ is
	%\begin{equation}\label{eq:G-appendix}
	%G(y, z)=\ii p(t; y, z)\diff t=(s(y)-s(\ell))\wedge (s(z)-s(\ell)), \quad y, z\in\mathcal{I},
	%\end{equation} where $p(t; y, z)$ is the transition density with respect to the speed measure.
	%Let us recall the other minimal excessive function
	%\begin{equation}\label{eq:kx-appendix}
	%k_x(y)=\min\left(k_{\ell}(y),\dfrac{k_r(y)}{s(x)-s(\ell)}\right)=\p_y(\text{hit}\; x \; \text{before}\; \ell)=\p_y(\lambda_x>0)=\begin{cases} \dfrac{s(y)-s(\ell)}{s(x)-s(\ell)}, \quad y\leq x,\\
	%\qquad 1, \quad \quad \quad \; y>x.
	%\end{cases}\end{equation}
	%The $k_x$-transform of $X$ is a regular diffusion with transition density function
	%\[P^{k_x}(t; u, A)=\frac{1}{k_x(u)}\int_A k_x(v) P(t; u, \diff v)\]
	%for a Borel set $A$. We reproduce the important identity \eqref{eq:identity-pre} here:
	%\begin{align}\label{eq:identity-appendix}
	%&\text{The $k_x$-transform (or $k_x$-diffusion) is identical in law with $X$}\\
	%&\text{when conditioned to hit $x$ and be killed at its last exit time from $x$.} \nonumber
	%\end{align}
	%That is, for all $x\in \mathcal{I}$, $\p_y^{k_x}(\lim_{t\rightarrow \zeta} X_t=x)=1$.
	First, we are interested in the case when $X$ starts at $y>x$. Then, $X$ would hit $x$ a.s. and the last exit time distribution from $x$ and the lifetime distribution of the abovementioned $k_x$-transform coincide. Thus, we wish to compute the lifetime distribution of $k_x$-diffusion.
	
	Let us consider the diffusion in space-time. Choose a point $(0, y)$ as a reference point in time and space. We set the Martin function (i.e., the minimal space-time excessive function: see Proposition 4 in Salminen \cite{salminen1984}) with support at a point $(t_0, x_0)$ as
	\begin{align*}
	_{(0, y)}k_{(t_0, x_0)}(v, z)=\begin{cases}
	\frac{p(t_0-v; z, x_0)}{p(t_0; y, x_0)}, & v<t_0,\\
	0, & v\ge t_0.
	\end{cases}
	\end{align*} We claim that, for $x<y$,
	\begin{equation}\label{eq:density}
	\p_y^{k_x}(\zeta \in \diff t)=\frac{p(t; y, x)}{s(x)-s(\ell)}\diff t=\frac{p(t; y, x)}{G(y, x)}\diff t.
	\end{equation}
	If we integrate the minimal space-time excessive function along the line $x=x_0$ with respect to the right-hand side of \eqref{eq:density}, we should have a space-time excessive function $h$ such that $h(0, y)=1$.
	Indeed,
	\begin{align*}
	h(v, z)&=\int_{0}^{\infty}{_{(0, y)}k_{(t, x_0)}(v, z)}\frac{p(t; y, x_0)}{G(y, x_0)}\diff t
	=\int_v^\infty \frac{p(t-v; z, x_0)}{p(t; y, x_0)}\frac{p(t; y, x_0)}{G(y, x_0)}\diff t\\
	&=\frac{G(z, x_0)}{G(y, x_0)}
	=\begin{cases}
	\dfrac{s(z)-s(\ell)}{s(x_0)-s(\ell)}, & z\le x_0,\\
	1, & z>x_0
	\end{cases}
	\end{align*} which is exactly
	$k_{x_0}(z)$.
	This shows that  the $k_x$-transform and the $h$-transform in the real line have the same finite dimensional distribution.  Therefore, we have $\p_y^{k_x}(\zeta \in \diff t)=\dfrac{p(t; y, x)}{G(y, x)}\diff t$ and
	\[
	\p_y^{k_x}(\zeta\in \diff t)=\p_y(\lambda_x\in \diff t)=\frac{p(t; y, x)}{G(y, x)}\diff t=\frac{p(t; y, x)}{s(x)-s(\ell)}\diff t \qquad \text{for} \; \: y>x
	\]  by \eqref{eq:G} and \eqref{eq:identity-pre}.

	Now let $y<x$. We need the following lemma, which is of interest in its own right:
	\begin{lemma}\label{lem:lifetime}
		Let $y<x$ and consider diffusion $X$ in Case 1. The $k_x$-transform of $X^{k_r}$ ($k_r$-transform of $X$) and the $k_x$-transform of $X$ starting from $y$ are identical in law.
	\end{lemma}
	
	\begin{proof}[Proof of Lemma \ref{lem:lifetime}]
		Recall that $k_r(x)=s(x)-s(\ell)$ and the transition density function of $k_r$-diffusion with respect to its speed measure is given by
		\begin{equation}\label{eq:kr-density-case1}
		p^{k_r}(t;y,x)m^{k_r}(\diff x)=\frac{p(t;y,x)k_r(x)}{k_r(y)}m(\diff x).
		\end{equation}
		Then, the scale function $s^r$ and the speed measure $m^r$ of $X^{k_r}$ are written as
		\begin{equation*}\label{eq:kr-parameters}
		s^r(x)=-\frac{1}{s(x)-s(\ell)} \conn m^r(\diff x)=(s(x)-s(\ell))^2 m(\diff x) \end{equation*}
		and the transition density, with respect to $m^r$, takes the form
		\begin{equation*}\label{eq:kr-density-evaluated}
		p^{k_r}(t;y,x)=\frac{p(t;y,x)}{(s(y)-s(\ell))(s(x)-s(\ell))},
		\end{equation*} which we write as $p^r(t; y, x)$ for short.
		We have
		\begin{align}\label{eq:kx-transforms}
		(\p^{k_r}_y)^{k_x}(t;y,A)&=\int_Ap^r(t;y,u)m^r(\diff u)\frac{\p_u^r(\lambda_x>0)}{\p_y^r(\lambda_x>0)}=\int_Ap(t;y,u)m(\diff u)\frac{s(u)-s(\ell)}{s(y)-s(\ell)}\frac{\p^r_u(\lambda_x>0)}{1}\nonumber\\
		&=\begin{cases} \bigintssss_Ap(t;y,u)m(\diff u)\dfrac{s(u)-s(\ell)}{s(y)-s(\ell)}, \quad u<x\\\nonumber
		\bigintssss_Ap(t;y,u)m(\diff u)\dfrac{s(u)-s(\ell)}{s(y)-s(\ell)}\dfrac{s^r(r)-s^r(u)}{s^r(r)-s^r(x)}, \quad u\geq x\end{cases}\\\nonumber
		&=\begin{cases} \bigintssss_Ap(t;y,u)m(\diff u)\dfrac{s(u)-s(\ell)}{s(y)-s(\ell)}, \quad u<x\\\nonumber
		\bigintssss_Ap(t;y,u)m(\diff u)\dfrac{s(x)-s(\ell)}{s(y)-s(\ell)}, \quad u\geq x\end{cases}\\\nonumber
		&=\int_Ap(t;y,u)m(\diff u)\frac{\p_u(\lambda_x>0)}{\p_y(\lambda_x>0)}\\
		&=\p^{k_x}_y(t;y,A).
		\end{align}
		In the first line, we used the definition of $k_x(\cdot)$ in \eqref{eq:kx-case1} and \eqref{eq:kr-density-case1} with $k_r(\cdot)=s(\cdot)-s(\ell)$.  For the third line, we used  $s(r)=+\infty$ in computing $s^{r}(r)=0$.
	\end{proof}
	From Lemma \ref{lem:lifetime} and \eqref{eq:identity-pre}, we obtain that the lifetime of the $k_x$- transform and the last exit time from $x$ for the $k_r$-transform have the same distribution:
	\[\p_y^{k_r}(\lambda_x\in \diff t)=\p^{k_x}_y(\zeta\in \diff t).\]	
	Now  $X^{k_r}$ converges to $r$ a.s. and starting from $y<x$, $X^{k_r}$ visits level $x$ a.s.. Hence, we can argue as in the previous case (i.e., $x<y$) with $X$ replaced by $X^{k_r}$. In particular, the last passage time to $x$ has the distribution
	\[\p^{k_r}_y(\lambda_x\in\diff t)=\p^{k_x}_y(\zeta\in \diff t)=\frac{p^{r}(t;y;x)}{G^r(y,x)}\diff t=\frac{p^{r}(t;y;x)}{s^r(r)-s^r(x)}\diff t=\frac{p(t;y,x)}{s(y)-s(\ell)}\diff t.\]
	
	In sum, for $y<x$, the last passage time to $x$ for our original $X$ has atom at $0$, since $X$ may not hit $x$ at all. The continuous part is given by
	\begin{equation}\label{eq:cont-part}
	\begin{aligned}
	\p_y(\lambda_x\in \diff t,\lambda_x>0)&=\p_y(\lambda_x\in \diff t\mid\lambda_x>0)\p_y(\lambda_x>0)=\p^{k_x}_y(\zeta\in\diff t)\times \p_y(T_x<T_{\ell})\\
	&=\frac{p(t;y,x)\diff t}{s(y)-s(\ell)}\frac{s(y)-s(\ell)}{s(x)-s(\ell)}=\frac{p(t;y,x)}{s(x)-s(\ell)}\diff t.
	\end{aligned}
	\end{equation}

	\noindent \underline{Case 2:}
	Let us assume that the scale function for the diffusion $X$ satisfies $s(\ell)>-\infty$ and $s(r)<\infty$. This includes the case when the left boundary is regular with killing and the right boundary is natural.
	Recall the minimal excessive function \eqref{eq:kx-case2}
	%\begin{equation}\label{eq:kl-appendix} k_{\ell}(x)=s(r)-s(x), \quad k_r(x)=s(x)-s(\ell), \end{equation}
	%and
	\begin{equation}\nonumber
	k_x(y)=
	\begin{cases}\dfrac{s(y)-s(\ell)}{s(x)-s(\ell)}, \quad y\leq x,\\
	\dfrac{s(r)-s(y)}{s(r)-s(x)}, \quad y>x,
	\end{cases}
	\end{equation}
	with $x, y\in\mathcal{I}$.
	%The Green function for $X$ is
	%\begin{equation*}\label{eq:G1-appendix}
	%G(y, z)=\ii p(t; y, z)\diff t=
	%\begin{cases}
	%\dfrac{(s(y)-s(\ell))(s(r)-s(z))}{s(r)-s(\ell)}\quad \ell<y\leq z<r,\\
	%\dfrac{(s(z)-s(\ell))(s(r)-s(y))}{s(r)-s(\ell)}\quad \ell<z\leq y<r,\\
	%\end{cases}
	%\end{equation*} where $p(t; y, z)$ is the transition density with respect to the speed measure.
	We observe that $k_x(y)=\p_y(\lambda_x>0)$. The diffusion $X$ in this case may not visit $\ell$ in finite time, so we force it to do so: the condition that $\lim\limits_{t\to\zeta}X_t=\ell$ is equivalent to $T_{\ell}<T_r$. By the Markov property, the transition density of this conditioned diffusion with respect to the speed measure is
	\begin{align*}
	p^*(t;u,v)m^*(\diff v)&=\p_u(X_t\in \diff v \mid T_{\ell}<T_r)=\dfrac{\p_u(X_t\in \diff v, T_{\ell}<T_r)}{\p_u(T_{\ell}<T_r)}=\dfrac{p(t;u,v)m(\diff v)\p_v(T_{\ell}<T_r)}{\p_u(T_{\ell}<T_r)}\\
	&=p(t;u,v)m(\diff v)\dfrac{s(r)-s(v)}{s(r)-s(\ell)}\div\dfrac{s(r)-s(u)}{s(r)-s(\ell)}=p(t;u,v)m(\diff v)\dfrac{s(r)-s(v)}{s(r)-s(u)},
	\end{align*}
	where $m$ denotes the speed measure of the original diffusion $X$. Note that this is the same as the $k_{\ell}$-transform (see \eqref{eq:minimal-case2}). The transition density of the conditioned diffusion $X^*$ with respect to its speed measure is thus given by \[p^*(t;u,v)=p(t;u,v)\dfrac{m(\diff v)}{m^*(\diff v)}\dfrac{s(r)-s(v)}{s(r)-s(u)}.\]
	Using a Taylor expansion of scale functions (see \citet[Chapter 15, Sec. 9]{karlin-book}), we obtain the scale function and the speed density for $X^*$:
	\begin{equation}\label{eq:s-star-appendix}
	s^*(x)=\dfrac{1}{s(r)-s(x)} \qquad m^*(\diff x)=(s(r)-s(x))^2m(\diff x)
	\end{equation}
	and
	\begin{equation}\label{eq:p-star-appendix}
	p^*(t;u,v)=\dfrac{p(t;u,v)}{(s(r)-s(v))(s(r)-s(u))}.
	\end{equation}
	We now consider the $k_x$-transform of $X^*$. Again, this transform is identical in law with $X^*$ when conditioned to hit $x$ and be killed at its last exit time from $x$. Since $s^*(r)=+\infty$ and $s^*(l)>-\infty$, $X^*$ behaves  similarly to the diffusion $X$ in Case 1.
	
	When $y>x$, $X^*$ would almost surely hit $x$ and the lifetime distribution of its $k_x$-transform and the last exit time distribution from $x$ coincide (see \eqref{eq:identity-pre}). From the result of Case 1, we have
	\[\p_y^*(\lambda_x\in \diff t)=\dfrac{p^*(t;y,x)}{G^*(y,x)}\diff t=\dfrac{p^*(t;y,x)}{s^*(x)-s^*(\ell)}\diff t=\dfrac{p(t;y,x)(s(r)-s(\ell))}{(s(r)-s(y))(s(x)-s(\ell))}\diff t\]
	where we used \eqref{eq:G}, \eqref{eq:s-star-appendix}, and \eqref{eq:p-star-appendix}.
	Now, $\p_y(T_{\ell}<T_r)=\dfrac{s(r)-s(y)}{s(r)-s(\ell)}$ and the last exit time distribution when $\lim\limits_{t\to\zeta}X_t=\ell$ is given by
	\begin{align*}
	\p_y(\lambda_x\in \diff t, T_{\ell}<T_r)&=\p_y(\lambda_x\in\diff t\mid T_{\ell}<T_r)\times\p_y(T_{\ell}<T_r)=\p_y^*(\lambda_x\in \diff t)\times \p_y(T_{\ell}<T_r)=\dfrac{p(t;y,x)}{s(x)-s(\ell)}\diff t
	\end{align*}
	and the result is the same as in Case 1.
	
	Next, let $y<x$. We consider the $k_r$-transform of the conditioned diffusion $X^*$. Since the conditioned diffusion $X^*$ behaves similarly to the diffusion $X$ in Case 1, we proceed similarly for $y<x$ in Case 1: its last passage time from $x$ has atom at $0$ and we get the density of the continuous part from \eqref{eq:cont-part}:
	\[\p_y^*(\lambda_x\in \diff t,\lambda_x>0)=\frac{p^*(t;y,x)}{s^*(x)-s^*(\ell)}\diff t=\frac{p(t;y,x)(s(r)-s(\ell))}{(s(r)-s(y))(s(x)-s(\ell))}\diff t.\]
	Finally,
	\[\p_y(\lambda_x\in \diff t,\lambda_x>0, T_{\ell}<T_r)=\p_y^*(\lambda_x\in \diff t,\lambda_x>0)\p_y(T_{\ell}<T_r)=\frac{p(t;y,x)}{s(x)-s(\ell)}\diff t\]
	and it is the same as in Case 1.
\end{proof}
For later reference, we record some important equations derived in the proof of Proposition \ref{prop:1}.
\begin{remark}\label{rmrk-x-star}
	\begin{normalfont}
		The scale function $s^r$ and the speed measure $m^r$ of the $k_r$-transform $X^{k_r}$ are written as
		\begin{equation}\label{eq:kr-parameters-pre}
		s^r(x)=-\frac{1}{s(x)-s(\ell)} \conn m^r(\diff x)=(s(x)-s(\ell))^2 m(\diff x) \end{equation}
		and the transition density, with respect to $m^r$, takes the form
		\begin{equation}\label{eq:kr-density-evaluated-pre}
		p^{k_r}(t;y,x)=\frac{p(t;y,x)}{(s(y)-s(\ell))(s(x)-s(\ell))},
		\end{equation} which we write as $p^r(t; y, x)$, in short.
		In Case 2, the diffusion conditioned that the killing occurs in finite time $T_{\ell}<T_r$ is denoted by $X^*$. Its
		scale function and speed density are
		\begin{equation}\label{eq:star-parameters}
		s^*(x)=\dfrac{1}{s(r)-s(x)} \conn m^*(\diff x)=(s(r)-s(x))^2m(\diff x),\end{equation}
		respectively and the transition density, with respect to $m^*$, is given by
		\begin{equation*}\label{eq:star-density}
		p^*(t;u,v)=\dfrac{p(t;u,v)}{(s(r)-s(v))(s(r)-s(u))}.\end{equation*}
	\end{normalfont}
\end{remark}
\bigskip

As we have discussed in Section \ref{subsec:2-2}, for the leverage process, we are interested in the last passage time of Brownian motion (starting at $y$) with drift $\mu(\neq0)$ and unit variance parameter to the state $\alpha$:
\[\lambda_\alpha:=\sup\{t\ge 0: X_t=\alpha\},\]
which is $0$ if the set in the brackets is empty. The \emph{scale function} $s(\cdot)$ for such Brownian motion is given by
\begin{equation}\label{scale-function}
s(x)=\frac{1}{2\mu}\left(1-e^{-2\mu x}\right)
\end{equation}
for $x\in(c,\infty)$, which is the state space in our case (\citet[Appendix 1]{borodina-salminen}. The left boundary $c$ is attracting since $s(c)>-\infty$. The right boundary $\infty$ can be attracting ($s(\infty)<\infty$ when $\mu>0$) or non-attracting ($s(\infty)=\infty$ when $\mu<0$).
We now apply Proposition \ref{prop:1} to our model.
\begin{corollary}\label{coro:1}
	Let $c$ be a regular killing boundary and consider Brownian motion $X_t$ on the state space $\mathcal{I}=(c, \infty)$ with drift $\mu=\frac{\nu-r}{\sigma}\neq 0$ and unit variance parameter. For any $y,\alpha\in \mathcal{I}$ satisfying $\alpha<y$, the distribution of $\lambda_\alpha$ when the company becomes insolvent in finite time $\big(\lim\limits_{t\to\zeta}X_t=c\big)$ is
	\begin{equation}\label{prop:1-eq}
	\p_y(\lambda_\alpha\in \diff t)=\frac{p(t; y, \alpha)}{\frac{1}{2\mu}\left(e^{-2\mu c}-e^{-2\mu\alpha}\right)}\diff t
	\end{equation}  where $p(t; u, v)$ is the transition density of the Brownian motion (with drift $\mu$ and unit variance parameter) being killed at $c$
	\begin{align}\label{eq:original-transition-density}
	&p(t; u, v)=\dfrac{1}{2\sqrt{2\pi t}}\exp\left(-\mu(u+v)-\dfrac{\mu^2t}{2}\right)\times \left(\exp\left(-\dfrac{(u-v)^2}{2t}\right)-\exp\left(-\dfrac{(u+v-2c)^2}{2t}\right)\right)
	\end{align} for $u,v>c$ with respect to the speed measure $m(\diff v)=2e^{2\mu v}\diff v$.
	For $\alpha>y$, the distribution of $\lambda_\alpha$ when the company becomes insolvent in finite time has atom at $0$. The continuous part $\p_y(\lambda_{\alpha}\in\diff t,\lambda_{\alpha}>0)$ is given by \eqref{prop:1-eq}.
\end{corollary}
\begin{proof}[Proof of Corollary \ref{coro:1}]
	The negative drift $\mu<0$ corresponds to Case 1 and the positive drift $\mu>0$ corresponds to Case 2 in Proposition \ref{prop:1}. We simply read the proof with $\ell=c$.
\end{proof}
\begin{remark}
	$1)$ \eqref{eq:original-transition-density} is the equation $(2.1.8)$ in \citet[Corollary 2.1.10]{platen}, where $u,v>c$.
	\newline $2)$ When $\mu>0$, the company may not become insolvent in finite time, with probability \[\p_y(T_r<T_c)=\dfrac{s(y)-s(c)}{s(r)-s(c)}=1-e^{-2\mu(y-c)},\] where we have used $s(r)=s(\infty)=\frac{1}{2\mu}$.
\end{remark}
\subsection{Time Until Death After Last Passage Time}\label{subsec:4-general}
We want to find $\p_y\{(T_{\ell}-\lambda_\alpha) \in \diff t\}$ (or $\p_y\{(T_{\ell}-\lambda_\alpha) \in \diff t, T_{\ell}<\infty\}$ for case 2). However, since $T_{\ell}$ and $\lambda_\alpha$ are far from independent under $\p_y$, it is not easy to compute this distribution. To bypass this difficulty, we consider the reversed path of the original diffusion (or of the original diffusion conditioned on the event $T_{\ell}<\infty$ for case 2) from $\ell$ and its first passage time to state $\alpha$. Note that at $T_{\ell}$, the original diffusion (or the conditioned diffusion) hits the state $\ell$ and is killed at $\ell$.
\begin{proposition}\label{prop:reversal}
	Let $X$ be a time-homogeneous regular diffusion process starting from $y \in \mathcal{I}=(\ell, r)$, $\ell$ being a regular killing boundary. Suppose that (Case 1) $s(\ell)>-\infty$ and $s(r)=+\infty$, or (Case 2) $s(\ell)>-\infty$ and $s(r)<+\infty$.  Then, in either case, given that $\lim\limits_{t\to\zeta}X_t=\ell$, the reversed process is the $k_r$-transform starting from $\ell$. That is, the two processes $\{X_{T_{\ell}-t}, 0\leq t\leq T_{\ell}\}$ and $\{X^{k_r}_t,  0\leq t\leq L_y^r\}$ have the same law with $L_y^r=\sup\{t:X_t^{k_r}=y\}$.
\end{proposition}
\begin{proof}[Proof of Proposition \ref{prop:reversal}]
	\underline{Case 1: } %We have $X=\{\omega_t, t\ge 0, \p_x\}$ with $\mathcal{I}=(l, r)$; $l$ is a regular boundary with killing and $r$ is a natural boundary. $s(r)=\infty$ and $s(l)>-\infty$.
	We confirm that $X^{k_r}$ ($k_r$-transform starting at $\ell$) reversed at the last exit time to $y$ (i.e., $L_y^r$) is the original diffusion $X$. First, note that $k_r(x)=s(x)-s(\ell)$, and as we already have seen in the proof of the previous proposition, the transition density function of $X^{k_r}$ with respect to its speed measure is given by \eqref{eq:kr-density-evaluated-pre}, which we reproduce here
	\begin{equation*}
	p^{k_r}(t;y,x)=\frac{p(t;y,x)}{(s(y)-s(\ell))(s(x)-s(\ell))}.
	\end{equation*}\label{eq:kr-density}
	The scale function $s^r$ and the speed measure $m^r$ of the $X^{k_r}$ are also written as in \eqref{eq:kr-parameters-pre}:
	\begin{equation}\label{eq:kr-parameters-again}
	s^r(x)=-\frac{1}{s(x)-s(\ell)} \conn m^r(\diff x)=(s(x)-s(\ell))^2 m(\diff x).
	\end{equation}
	Note that $s^r(\ell)=-\infty$. In terms of $s^r$, the Green function is
	\begin{align}\label{eq:G*}
	G^r(y, z)=\begin{cases}
	s^r(r)-s^r(z), & y\le z,\\
	s^r(r)-s^r(y), & z\le y.
	\end{cases}
	\end{align}
	Recall that Nagasawa's theorem on time reversal in this context reads as follows: (see \citet{nagasawa} and \citet{sharpe1980})
	\begin{theorem}\label{thm:nagasawa}
		Let $X$ and $\widehat{X}$ be standard Markov processes in duality on their common state space $E$ relative to a $\sigma$-finite reference measure $\xi$. Let $u(x, y)$ be the potential kernel density relative to $\xi$ so that $Uf(x)=\E^x \ii f(X_t)\diff t=\int u(x, y) f(y)\xi(\diff y)$.  Let $L$ be a cooptional time for $X$; that is, $L\circ\theta(t):=(L-t)^+$, where $\theta(\cdot)$ is the shift operator. Denote $\widetilde{X}$ by
		\begin{align*}
		\widetilde{X}_t:=\begin{cases}
		X_{(L-t)-}, & \text{for $0<t<L$ on $\{0<L<\infty\}$},\\
		\Delta, &\text{otherwise}.
		\end{cases}
		\end{align*}
		For an initial distribution $\lambda$, let $v(y):=\int \lambda(\diff x)u(x, y)$. Then, under $\p^\lambda$, the reversed process $(\widetilde{X}_t)_{t>0}$ is a homogeneous Markov process on $E$ with transition semigroup $(\widetilde{P}_t)$ given by
		\begin{align*}
		\widetilde{P}_t f(y)=\begin{cases}
		\widehat{P}_t(fv)(y)/v(y), & 0<v(y)<\infty,\\
		0, &\text{otherwise}.
		\end{cases}
		\end{align*}
	\end{theorem}
	We use Theorem \ref{thm:nagasawa} to continue the proof of Proposition \ref{prop:reversal}. In our case, $X^{k_r}$ is self-dual relative to its speed measure $m^r$. %and the potential operator is $G^r$.
	Read the theorem with $E=[\ell,r)$, $\xi(\diff y)=m^r(\diff y)$, $L=L_y^r$, $u(x,y)=G^r(x,y)$, and set $\lambda(\diff x)=\delta_{\ell}(\diff x)$, which is the Dirac measure with total mass at $\ell$, to obtain
	\[v^r(z)=\int_{\ell}^{r} \delta_{\ell}(\diff y)G^r(y, z)=G^r(\ell, z)=s^r(r)-s^r(z)\]
	since $\ell<z, \forall z\in \mathcal{I}$ in \eqref{eq:G*}. For a nonnegative Borel measurable function $f$, the potential operator of $\tilde{X}^{k_r}$ (the reversed process of $X^{k_r}$), denoted by $V$, is written as
	\begin{align*}
	&Vf(y)=\frac{U(fv^r)(y)}{v^r(y)}=\int_{\ell}^{r} G^r(y,z)\frac{v^r(z)}{v^r(y)}f(z)m^r(\diff z)\\
	&=\int_{\ell}^y (s^r(r)-s^r(y))\frac{s^r(r)-s^r(z)}{s^r(r)-s^r(y)}f(z)m^r(\diff z)+\int_y^r (s^r(r)-s^r(z))\frac{s^r(r)-s^r(z)}{s^r(r)-s^r(y)}f(z)m^r(\diff z)\\
	&=\int_{\ell}^{r} G^{**}(y, z)f(z)(s^r(r)-s^r(z))^2m^r(\diff z)
	%%&=\int_{\mathcal{I}} G^{**}(y, z)f(z)m(\diff z)
	\end{align*} where
	\begin{align*}
	G^{**}(y, z):=\begin{cases}
	\frac{1}{s^r(r)-s^r(z)}, & z\le y, \\
	\frac{1}{s^r(r)-s^r(y)}, &y\le z.
	\end{cases}
	\end{align*}
	By \eqref{eq:kr-parameters-again} and our assumption $s(r)=+\infty$, we have $s^r(r)=0$, and this simplifies $Vf(y)$ to
	\begin{equation}\label{eq:potential-V}
	Vf(y)=\int_{E} \left(-\frac{1}{s^r(z)}\right)\wedge \left(-\frac{1}{s^r(y)}\right)f(z)s^r(z)^2m^r(\diff z).
	\end{equation}
	This reversed transform $\tilde{X}^{k_r}$ dies only at $\ell$. Now, take any $\ell<y<a<b<\infty$. Using the Markov property, we have
	\begin{align*}
	\nonumber
	V\1_{(a,b)}(y)&=\E_y\left[\int_0^{\infty}\1_{(a,b)}(\tilde{X}_t)\diff t\right]=\E_y\left[\int_{\tilde{T}_a}^{\infty}\1_{(a,b)}(\tilde{X}_t)\1_{\tilde{T}_a<\tilde{T}_{\ell}}\diff t\right]=\int_{\tilde{T}_a}^{\infty}\E_y\left[\1_{(a,b)}(\tilde{X}_t)\1_{\tilde{T}_a<\tilde{T}_{\ell}}\right]\diff t\\
	&=\int_{\tilde{T}_a}^{\infty}\E_y\left[\1_{\tilde{T}_a<\tilde{T}_{\ell}}\E_{\tilde{X}_{\tilde{T}_a}}\left(\1_{(a,b)}\left(\tilde{X}_{t-\tilde{T}_a}\right)\right)\right]\diff t=\int_{\tilde{T}_a}^{\infty}\E_y\left[\1_{\tilde{T}_a<\tilde{T}_{\ell}}\right]\E_a\left[\1_{(a,b)}\left(\tilde{X}_{t-\tilde{T}_a}\right)\right]\diff t\\
	&=\E_y\left[\1_{\tilde{T}_a<\tilde{T}_{\ell}}\right]V\1_{(a,b)}(a).
	\end{align*}
	Hence, $\E_y\left[\1_{\tilde{T}_a<\tilde{T}_{\ell}}\right]=\frac{V\1_{(a,b)}(y)}{V\1_{(a,b)}(a)}=\dfrac{-\frac{1}{s^r(y)}}{-\frac{1}{s^r(a)}}$ from \eqref{eq:potential-V} by noting that $y<a$. Indeed,
	\begin{align*}
	\frac{V\1_{(a,b)}(y)}{V\1_{(a,b)}(a)}&=\frac{\left(-\frac{1}{s^r(y)}\right)\int_\mathcal{I}\1_{(a,b)}(z)s^r(z)^2 m(\diff z)}{\left(-\frac{1}{s^r(a)}\right)\int_\mathcal{I}\1_{(a,b)}(z)s^r(z)^2 m(\diff z)},
	\end{align*}
	since $s^r(\cdot)$ is monotone increasing.  %Similarly, using $y>a>b$, we get $\E_y\left[\1_{\tilde{T}_a<\tilde{T}_l}\right]=1$.
	Now, by \eqref{eq:kr-parameters-pre}, we can write that for $\ell<y<a$,
	\begin{align*}
	\E_y\left[\1_{\tilde{T}_a<\tilde{T}_{\ell}}\right]=\frac{-\frac{1}{s^r(y)}}{-\frac{1}{s^r(a)}}=\frac{s(y)-s(\ell)}{s(a)-s(\ell)},
	\end{align*} which shows that the scale function of $\tilde{X}^{k_r}$ is $s(\cdot)$ and coincides with that of our original diffusion $X$. Also, $(s^r(z))^2m^r(\diff z)$ becomes the speed measure of $\tilde{X}^r$ by definition. Note that $(s^r(z))^2m^r(\diff z)=m(\diff z)$ as expected, and the speed measure coincides with the speed measure of the original diffusion. Since both $X$ and $\tilde{X}^{k_r}$ are killed only at $\ell$ and their scale functions and speed measures coincide, we conclude that the reversed process from $\ell$ of our original diffusion $X$ is its $k_r$-transform starting from $\ell$.
	
	\underline{Case 2:} Let us consider the case $s(\ell)>-\infty$ and $s(r)<+\infty$. We use the conditioned diffusion $X^*$ from the proof in Proposition \ref{prop:1} (see Remark \ref{rmrk-x-star}). For this diffusion, $s^*(r)=\infty$ and $s^*(\ell)>-\infty$. Therefore, we can use the same arguments as in the proof of Case 1, and we conclude that the $k_r$-transform of this conditioned diffusion started at $\ell$ and the conditioned diffusion itself are reversals of each other.
\end{proof}

By Proposition \ref{prop:reversal}, the original problem $\p_y\{(T_{\ell}-\lambda_\alpha) \in \diff t\}$ is now converted to the \emph{first passage time of $k_r$-diffusion} (starting at $\ell$) to level $\alpha$. While the first passage time distribution may not be always available, it is much more convenient than handling the joint density of $\lambda_\alpha$ and $T_\ell$ in the original problem.

For the leverage process in Section \ref{subsec:2-2}, we want to find $\p_y\{(T_c-\lambda_\alpha) \in \diff t\}$ (or $\p_y\{(T_c-\lambda_\alpha) \in \diff t, T_c<\infty\}$ when $\mu>0$). This distribution should be useful in knowing how long the firm would have for implementing its measures to avoid insolvency.
\begin{corollary}\label{coro:2}
	Let $X_t$ be a Brownian motion with drift $\mu\neq 0$ and unit variance on $\mathcal{I}=(c,\infty)$, with $c$ being a regular killing boundary. (a) When $\mu<0$, the reversed process of $X_t$ from $L_c=\inf\{t : X_t=c\}$ has the generator $\tilde{\G}$
	\begin{equation}\label{eq:generator}
	\tilde{\G}f(x)=\frac{\mu\left(e^{2\mu(x-c)}+1\right)}{e^{2\mu(x-c)}-1}f'(x)+\frac{1}{2}f''(x)
	\end{equation} and the transition density with respect to speed measure $\tilde{m}(\diff v)=\left(e^{-2\mu c}-e^{-2\mu v}\right)\left(\frac{1}{2\mu}\right)^22e^{2\mu v}\diff v$,
	\[\tilde{p}(t; u, v)=\frac{\dfrac{1}{2\sqrt{2\pi t}}\exp\left(-\mu(u+v)-\dfrac{\mu^2t}{2}\right)\times \left(\exp\left(-\dfrac{(u-v)^2}{2t}\right)-\exp\left(-\dfrac{(u+v-2c)^2}{2t}\right)\right)}{\left(e^{-2\mu c}-e^{-2\mu v}\right)\left(e^{-2\mu c}-e^{-2\mu u}\right)\left(\frac{1}{2\mu}\right)^2}\]
	for $u,v>c$ and
	\[\tilde{p}(t; c, v)=\frac{2\mu}{\sqrt{2\pi t}}\frac{e^{-\mu(v-c)-\frac{\mu^2t}{2}}\frac{v-c}{t}e^{-\frac{(v-c)^2}{2t}}}{e^{-2\mu c}-e^{-2\mu v}}\] for $u=c$.
	Note also that the distribution of $T_c-\lambda_\alpha$ under $\p_y$ is the same as that of $T_{\alpha}:=\inf\{t\ge 0: \omega(t)=\alpha\}$ under $\tilde{\p}_c$.\\
	(b) When $\mu>0$, the reversed process of $X_t^*$ ($X_t$ conditioned to hit $c$ in finite time and be killed at time $T_c$) from $L^*_c=\inf\{t : X^*_t=c\}$ has the generator $\tilde{\G}$
	\begin{equation*}\label{eq:generator-2}
	\tilde{\G}f(x)=\frac{\mu\left(e^{2\mu(x-c)}+1\right)}{e^{2\mu(x-c)}-1}f'(x)+\frac{1}{2}f''(x)
	\end{equation*} and the transition density with respect to speed measure $\tilde{m}(\diff v)=\frac{e^{-2\mu c}-e^{-2\mu v}}{e^{-4\mu c}}2e^{2\mu v}\diff v$
	\[\tilde{p}(t; u, v)=\frac{\dfrac{1}{2\sqrt{2\pi t}}\exp\left(-\mu(u+v)-\dfrac{\mu^2t}{2}\right) \left(\exp\left(-\dfrac{(u-v)^2}{2t}\right)-\exp\left(-\dfrac{(u+v-2c)^2}{2t}\right)\right)e^{-4\mu c}}{\left(e^{-2\mu c}-e^{-2\mu u}\right)\left(e^{-2\mu c}-e^{-2\mu v}\right)}\]
	for $u,v>c$ and
	\[\tilde{p}(t; c, v)=\frac{1}{2\mu\sqrt{2\pi t}}e^{-\mu(c+v)-\frac{\mu^2t}{2}}\frac{v-c}{t}e^{-\frac{(v-c)^2}{2t}}\frac{1}{\left(1-e^{-2\mu(v-c)}\right)}\] for $u=c$.
	Note also that the distribution of $T_c-\lambda_\alpha$ under $\p_y$ when $T_c<\infty$ is the same as that of $T_{\alpha}:=\inf\{t\ge 0: \omega(t)=\alpha\}$ under $\tilde{\p}_c$ multiplied by $\p_y(T_c<\infty)$.
\end{corollary}
\begin{proof}[Proof of Corollary \ref{coro:2}]
	As it can be seen from the proof of Proposition \ref{prop:reversal}, if we reverse the original process $X$ (or $X^*$) from $L_c=\inf\{t : X_t=\ell\}$, we obtain the $k_r$-transform.
	
	\underline{(a) $\mu<0$:  } This corresponds to Case 1 in Proposition \ref{prop:reversal}. We are interested in finding the $k_r$-transform of Brownian motion with drift $\mu<0$ killed at $c$. From \eqref{eq:kr-parameters-pre}, we obtain the following with \eqref{scale-function}:
	\[ s^r(x)=-\frac{1}{\frac{1}{2\mu}(e^{-2\mu c}-e^{-2\mu x})} \conn m^r(\diff x)=\underset{:=\bar{m}^r(x)}{\underbrace{\left(\frac{1}{2\mu}(e^{-2\mu c}-e^{-2\mu x})\right)^22e^{2\mu x}}}\diff x.\]
	
	Since $\tilde{\G} f(x)=\frac{1}{\bar{m}^r(x)}\frac{\diff}{\diff x}\left(\frac{1}{(s^r(x))'}\frac{\diff f(x)}{\diff x}\right)$ (e.g., see \citet[Chapter 15, Sec.3]{karlin-book}), we simply use $s^r(\cdot)$ and $\bar{m}^r(\cdot)$ to obtain \eqref{eq:generator}. The transition density with respect to its speed measure is obtained by \eqref{eq:kr-density-evaluated-pre}. The entrance law from $c$ is due to L'H\^{o}pital's rule. Note that the drift parameter becomes very large when the $k_r$-diffusion approaches $c$, so that the process stays away from $c$. This is also confirmed by the fact that $s^r(c)=-\infty$: once the $k_r$-diffusion enters from $c$, it can never reach $c$. By this fact and the definition of $\lambda_\alpha$, we have the final assertion.
	
	\underline{(b) $\mu>0$: } This corresponds to Case 2 in Proposition \ref{prop:reversal}. Here, we are interested in the $k_r$-transform of the Brownian motion with drift $\mu>0$ conditioned to hit $c$ in finite time. This conditional diffusion is the same as $X^*$ from the proof of Proposition \ref{prop:1}. Use the equation for $s^r(x)$ in \eqref{eq:kr-parameters-pre} but replace $s(x)$ with $s^*(x)$ in \eqref{eq:star-parameters} to obtain \[s^r(x)=-\frac{e^{-2\mu(x+c)}}{2\mu(e^{-2\mu c}-e^{-2\mu x})} \conn m^r(\diff x)=\underset{:=\bar{m}^r(x)}{\underbrace{\left(1-e^{-2\mu(x-c)}\right)^22e^{2\mu x}}}\diff x.\]
	All other assertions are derived in the same way as in Case 1.
\end{proof}

In the case of a Brownian motion with drift, we obtain a semi-explicit expression of the density of the time left until killing (after the last passage time to $\alpha$) in terms of its Laplace transform. We can apply this result to leverage process in Section \ref{subsec:2-2}. Note that without loss of generality, we have assumed the unit variance in Proposition \ref{prop:additional} since it is possible to standardize variance to one by adjusting $\alpha$ of interest (see Section \ref{subsec:2-2}).
\begin{proposition} \label{prop:additional}
	Given that $T_c$ is finite, the time until insolvency after the last passage time to $\alpha$ for the Brownian motion $X_t$ with drift $\mu=\frac{\nu-r}{\sigma}\neq 0$ and unit variance is given by
	\begin{equation}\label{eq:answer-to-problem2}
	\p_.(T_c-\lambda_\alpha\in \diff t)=\lim_{x\downarrow c}\frac{\sinh(\mu(\alpha-c))}{\sinh(\mu(x-c))}\cdot  e^{-\frac{1}{2}\mu^2 t}\p_x(H_\alpha\in \diff t, t<H_c),
	\end{equation} where
	\[
	H_\alpha:=\inf\{t\ge 0: B_t=\alpha\}.
	\] $B=(B_t)_{t\ge 0}$ is a standard Brownian motion under $\p_x$, and
	$\p_x(H_\alpha\in \diff t, t<H_c)$ has the Laplace transform
	\begin{equation*}%\label{eq:Laplace-for-prop}
	\int_0^\infty e^{-qt}\p_x(H_\alpha\in \diff t, t<H_c)=\frac{\sinh (\sqrt{2q}(x-c))}{\sinh(\sqrt{2q}(\alpha-c))}, \quad q>0.
	\end{equation*}
\end{proposition}

\begin{proof}[Proof of Proposition \ref{prop:additional}]
	First, note that the drift in \eqref{eq:generator} is $\dfrac{\mu\left(e^{2\mu(x-c)}+1\right)}{e^{2\mu(x-c)}-1}=\mu\coth(\mu(x-c))$, taking values in $\R_+$  for any $\mu\in \R\backslash\{0\}$ and $x>c$. The state space of $X$ is $\mathcal{I}=(c, \infty)$. By taking a point $x\in (c, \alpha)$, our reversed diffusion (see Corollary \ref{coro:2}) is
	\[
	\tilde{X}_t=X_0 +  \int_0^t \mu\coth(\mu(\tilde{X}_s-c)) \diff s + W_t, \quad X_0=x, \quad t\ge 0,
	\] where $W_t$ is a standard Brownian motion. This is the $k_r$-transform of $X$ conditioned on the event $\{T_c<\infty\}$. Following the result in Proposition \ref{coro:2}, we wish to compute the first passage time of $\tilde{X}$ (starting from $c$) to level $\alpha$, which is equal to $T_c-\lambda_\alpha$ for $X$ conditioned on $\{T_c<\infty\}$. For this, we use the following method of measure change.
	\newline\indent We take any standard Brownian motion $B_t$ defined on $(\Omega,\mathcal{F},\p)$ with $B_0=x$ a.s.. If necessary, we enlarge our original filtration to make $B$ adapted to that larger filtration. The following measure change works until any finite time $t$; therefore, we can choose to look at the path of $B_t$ until it hits $c$ for the first time, that is, $t<T_c^B$. We set
	\begin{equation}\label{eq:Z}
	Z_t=Z_t(B):=\exp\left(\int_0^t \mu\coth(\mu(B_s-c))\diff B_s -\frac{1}{2}\int_0^t \mu^2\coth^2(\mu(B_s-c))\diff s \right).
	\end{equation}
	Since $b(x)=\mu\coth(\mu(x-c))$ is bounded on $x\in [c+\epsilon, \infty)$ for any $\epsilon>0$ and any $\mu \in \R\backslash\{0\}$, the non-explosion condition (see \citet[Sec.9.4]{chung-williams}) is satisfied: there exists a positive constant $C$ such that
	\[
	x\cdot b(x)\le C(1+x^2), \quad x\in (c, \infty).
	\]  Then, $Z_t(B)$ is a martingale for $0\le t <T_c^B$ because $B$ never touches state $c$  for this time period once we start it at $x\in (c, \infty)$. For each $0\leq t<T_c^B$, we define a new probability measure $\widehat{\p}$ by the Radon--Nikodym derivative $Z_t=\dfrac{\diff \widehat{\p}_x}{\diff\p_x}\Bigr|_{\mathcal{F}_t}$. That is, for each $0\le S <T_c^B$, a probability measure $\widehat{\p}_x$ on $\F_S$ is given by
	\[
	\widehat{\p}_x(A):=\E_x[\1_A Z_S(B)], \quad A\in \F_S.
	\] By the Girsanov theorem, under the new measure,
	\[
	Y_t:=B_t-\int_0^t \mu\coth(\mu(B_s-c))\diff s
	\] is a standard Brownian motion, and $B_t$ is a Brownian motion with drift $b(B_s)=\mu\coth(\mu(B_s-c))$. Note that $B_t$ has the same state space $(c, \infty)$ and the same generator under $\widehat{\p}_x$ as $\tilde{X}_t$ under $\p_x$. Therefore, the distribution of $H_\alpha':=\inf\{t\ge 0: \tilde{X}_t=\alpha\}$ under $\p_x$ and $H_\alpha:=\inf\{t\ge 0: B_t=\alpha\}$ under $\widehat{\p}_x$ are the same:
	\begin{align}\label{eq:interim}
	&\widehat{\E}_x(\1_{H_{\alpha}\le t})=\widehat{\E}_x(\1_{H_{\alpha}\le t}\1_{H_{\alpha}<H_c})=\E_x[\1_{\{H_\alpha\le t\}}\1_{\{H_\alpha<H_c\}}Z_t]\nonumber\\
	&=\E_x\left[\1_{\{H_\alpha\le t\}}\1_{\{H_\alpha<H_c\}}\exp\left(\int_0^t \mu\coth(\mu(B_s-c))\diff B_s -\frac{1}{2}\int_0^t \mu^2\coth^2(\mu(B_s-c))\diff s \right)\right].
	\end{align} To simplify the above probability, define
	\begin{equation*}
	g(x)=\begin{cases}
	\ln\left( \sinh (-\mu(x-c))\right), \quad \mu<0 \quad x>c,\\
	\ln\left( \sinh (\mu(x-c))\right), \hspace{0.7cm}  \mu>0 \quad x>c.
	\end{cases}
	\end{equation*}  Then, in either case, we have for $x>c$
	\begin{align}\label{eq:derivatives-of-g}
	g'(x)=\mu \coth (\mu (x-c)) \conn g''(x)=\mu^2[1-\coth^2(\mu(x-c))].
	\end{align} By the It\^{o} formula, we have using \eqref{eq:derivatives-of-g}
	\begin{align*}
	&\int_0^t \mu\coth(\mu(B_s-c))\diff B_s\\
	&=\ln\left( \sinh (-\mu(B_t-c))\right)-\ln\left( \sinh(-\mu(x-c))\right)-\frac{1}{2}\int_0^t \mu^2[1-\coth^2(\mu(B_s-c))]\diff s
	\end{align*}when $\mu<0$ and
	\begin{align*}
	&\int_0^t \mu\coth(\mu(B_s-c))\diff B_s\\
	&=\ln\left( \sinh (\mu(B_t-c))\right)-\ln\left( \sinh(\mu(x-c))\right)-\frac{1}{2}\int_0^t \mu^2[1-\coth^2(\mu(B_s-c))]\diff s
	\end{align*}
	when $\mu>0$. Plugging this equation into \eqref{eq:interim} leads to a cancelation, and thus, we have
	\begin{align*}
	\widehat{\p}_x(H_\alpha\le t)&=[\sinh(\mu(x-c))]^{-1}\E_x\left[\1_{\{H_\alpha\le t\}}\1_{\{H_\alpha<H_c\}}\sinh(\mu(B_t-c))e^{-\frac{1}{2}\mu^2 t}\right]\\
	&=[\sinh(\mu(x-c))]^{-1}\E_x\left[\1_{\{H_\alpha\le t\}}\1_{\{H_\alpha<H_c\}}\sinh(\mu(\alpha-c))e^{-\frac{1}{2}\mu^2 H_\alpha}\right]
	\end{align*}
	where the last line is due to the optional sampling theorem for the martingale  \eqref{eq:Z} in view of \eqref{eq:interim}.  Write the density
	\begin{equation}\label{eq:p-hat-H}
	\widehat{\p}_x(H_\alpha\in \diff t)=M\cdot  e^{-\frac{1}{2}\mu^2 t}\p_x(H_\alpha\in \diff t, t<H_c), \quad M:=\frac{\sinh(\mu(\alpha-c))}{\sinh(\mu(x-c))}
	\end{equation} and compare with
	\begin{equation}\label{eq:Laplace}
	\E_x\left[e^{-qH_\alpha}\1_{\{H_\alpha<H_c\}}\right]=\frac{\sinh (\sqrt{2q}(x-c))}{\sinh(\sqrt{2q}(\alpha-c))}, \quad q>0
	\end{equation} (see e.g., \citet[Sec. 8.2]{Kyprianou_2014}) to see that $\p_x(H_\alpha\in \diff t, t<H_c)$ has the Laplace transform \eqref{eq:Laplace} with $q=\frac{1}{2}\mu^2$.
	%\begin{equation*}
	% \p_x(T_c-\lambda_\alpha\in \diff t)=M\cdot  e^{-\frac{1}{2}\mu^2 t}\p_x(H_\alpha\in \diff t, t<H_c)
	%\end{equation*} where $\p_x(H_\alpha\in \diff t, t<H_c)$ has the Laplace transform \eqref{eq:Laplace}.
	Note that we can confirm $\int_0^\infty\widehat{\p}_x(H_\alpha\in \diff t)=1$ for any $x\in (c, \infty)$ by integrating the right--hand side of \eqref{eq:p-hat-H} and using \eqref{eq:Laplace} with $q=\frac{1}{2}\mu^2$.  Finally, we use Proposition \ref{prop:reversal} to obtain \eqref{eq:answer-to-problem2}.
\end{proof}

\subsection{Endogenizing the Threshold}\label{sec:endogenizing}
In the previous sections, we calculated some functionals that involve $\lambda_\alpha$, the last visit to state $\alpha$ before the original process is killed. In this section, we wish to make level $\alpha$ endogenous: we obtain  this threshold  as a solution to a certain appropriate  optimization problem. Let us consider again a general diffusion $X$ taking values in $(\ell, r)$, with $\ell$ being a regular killing boundary and $r$ a natural boundary.

To form an optimization problem, it should be reasonable to assume the following:
\begin{enumerate}[leftmargin=*]
	\item If the level of $\alpha$ is too low, $X$ may hit $\ell$ shortly after it is below $\alpha$. From the perspective of risk management, this means that the firm may become insolvent shortly after it finds itself below the precautionary threshold. In this case, the management has missed out on a bad sign on a timely basis. Hence, the management wants an alarm early enough to implement some measures.
	\item The company as a whole wants to minimize the time spent below a precautionary threshold
	\[A_t:=\int_0^{T_{\ell}\wedge T_r\wedge t} \1_{(\ell, \alpha)}(X_s)\diff s.\]
	The creditors naturally want the company to operate above the precautionary threshold.  From the shareholders' point of view as well, the value of the investment in the company has decreased and they are at risk of losing the whole investment while the firm is operating below level $\alpha$. Also, it will be harder to receive dividends. The creditors and shareholders may have different points of view regarding the appropriate management strategies during the financial distress. We will discuss on behalf of which stakeholder the management acts later in Section \ref{subsubsec:alpha-effect} where we consider specific strategies.
\end{enumerate}
We here note that the parameter $\Gamma\in[0,1]$ in the subsequent optimization problem will represent the relative importance between the two quantities associated with (1) and (2). Using $\Gamma$, one can adjust the priority of the two quantities in a flexible way (see Subsection \ref{subsubsec:gamma-effect}).

For (1), we can consider the following probability for fixed $t$:
\begin{equation}\label{eq:first_quantity}
\p_y\left(Q_t=\lambda_{\alpha}, X_t\in(\ell,\alpha) \right)=\int_{\ell}^{\alpha}\frac{s(\alpha)-s(z)}{s(\alpha)-s(\ell)}\p_y(X_t\in\diff z),
\end{equation}
where $Q_t:=\sup\{s<t: X_s=\alpha\}$ (see Sections \ref{subsubsec:method} and \ref{subsubsec:evaluation}). The equality holds due to strong Markov property. This quantity is in line with $P_y(\lambda_{\alpha}>t|\F_t)$ in Remark \ref{rmrk:projection} while we are now considering the case $X_t<\alpha$ (here we replaced arbitrary $x$ in Remark \ref{rmrk:projection} with $\alpha$).  
The probability in \eqref{eq:first_quantity} indicates how likely it is that when the process is below $\alpha$ at time $t$, it will never recover to $\alpha$ and will hit the killing boundary. This quantity is an increasing function of $\alpha$. We naturally assume that at some point of time $t < T$, when the firm finds itself below the level $\alpha$ (i.e., $X_t\in(\ell,\alpha$)), it wants to minimize the probability that $Q_t$ is the last visit to $\alpha$ (i.e., $Q_t=\lambda_{\alpha}$). The firm can be sufficiently cautious, by setting alarming $\alpha$ at a high level (thus, raising the probability in \eqref{eq:first_quantity}), but the time spent below $\alpha$ (discussed in (2)) would be larger. In particular, for (2), we can consider the following Laplace transform
\begin{align*}\label{eq:opportunity-cost}
\E_y\left(e^{-q A_\infty}\right)&=\E_y\left(e^{-q \lim\limits_{t\to\infty}A_t}\right)=\E_y\left[e^{-q \int_0^{T_{\ell}\wedge T_r}\1_{(\ell, \alpha)}(X_s)\diff s}\right]\\
&=\lim_{t\rightarrow \infty} \E_y\left[e^{-q \int_0^{T_{\ell} \wedge T_r\wedge t}\1_{(\ell, \alpha)}(X_s)\diff s}\right],\qquad q>0,
\end{align*} where the equality holds due to the bounded convergence theorem. If we increase the level $\alpha$ (by making $\alpha$ closer to $y$), the quantity $A_\infty$ increases (and hence $\E_y\left(e^{-q A_\infty}\right)$ decreases).

Accordingly, for a given $t$, we set the optimization problem as a convex combination of the two terms:

\begin{equation}\label{eq:measure}
v(y; t):= \max_{\alpha\in [\ell, y]}\left[\Gamma\cdot\left(\p_y\left(Q_t=\lambda_{\alpha}, X_t\in(\ell,\alpha)\right)+\int_0^t\p_y\left(T_{\ell}\in\diff z\right)\right)+(1-\Gamma)\cdot\E_y\left(e^{-q A_\infty}, T_{\ell}<T_r\right)\right],
\end{equation}
where $\Gamma\in[0,1]$ indicates the relative importance between the two terms. The first term indicates the probability that up to time $t$, the process $X$ is already killed, or otherwise, $X_t$ is in $(\ell, \alpha)$ and it does not return to the level $\alpha$. From the risk management's point of view, we call the second term $\E_y\left(e^{-q A_\infty}, T_{\ell}<T_r\right)$ as \emph{financial distress}. This term is constructed based on the fact that our primary interest lies in the case where $T_{\ell}<T_r$. This corresponds to killing (insolvency) occurring in finite time. Note that
as for $\p_y\left(Q_t=\lambda_{\alpha}, X_t\in(\ell,\alpha)\right)$, the inclusion of the event $T_{\ell}<T_r$ in this probability would not make any difference. \\
\indent The equation \eqref{eq:measure} is constructed for a general diffusion $X$. We briefly summerize its implication for credit risk management. A higher $\alpha$ would increase the first probability term, thus increasing the sense of danger. This means that the management would be on alert and it will possibly try to avoid worsening the situation.
%\st{From the perspective of risk management, this is good, since in the most extreme case, where $\alpha=y$, the company would refrain from making new risky investments to make less volatile the company's asset value. }
However, the time spent below the precautionary threshold $A_{\infty}$ would be large as well, decreasing the second term. In summary, the management would want to set $\alpha$ to the level that would minimize the risk of insolvency; however, if $\alpha$ is too high, this would increase the time spent in financial distress. This is a tradeoff that we wanted to see.

\indent Now, $\E_y\left(e^{-q A_\infty}, T_{\ell}<T_r\right)$ is evaluated as (\citet[Sec. 4.1]{zhang2015-1} with $b\uparrow r$ and $a\downarrow \ell$)
\begin{equation*}\label{eq:zhang}
\E_y\left(e^{-q A_\infty},T_{\ell}<T_r\right)=
\begin{cases}
\begin{aligned}
\dfrac{(s(r)-s(y))s'(\alpha)}{(s(r)-s(\alpha))W_{q,1}(\alpha,\ell)+s'(\alpha)W_q(\alpha,\ell)}  \qquad &s(r)<\infty,\\
\hspace{1cm}\\
\dfrac{s'(\alpha)}{W_{q,1}(\alpha,\ell)} \qquad \qquad \qquad  \qquad & s(r)=\infty.\\
\end{aligned}
\end{cases}
\end{equation*}
$W_{q,1}$　is defined in \citet[Sec.2]{zhang2015-1} in the following manner. Let $\varphi$ and $\psi$ denote positive increasing and decreasing solutions of the o.d.e. $\G f(x)=qf(x)$ for $q>0$, where $\G$ is a generator of $X$. There exists a constant $\omega_q>0$ satisfying
\[\omega_q\cdot s'(x)=\varphi'(x)\psi(x)-\psi'(x)\varphi(x).\]
Setting
\[W_q(x,y)=\omega_q^{-1}\det\left[\begin{matrix}
\varphi(x) &\varphi(y)\\
\psi(x) & \psi(y)
\end{matrix}\right]\]
for $x,y\in\mathcal{I}$, $W_{q,1}(x,y)=\frac{\partial}{\partial x}W_q(x,y)$. \\
Thus, our optimization problem becomes
\begin{small}
	\begin{align}\label{eq:measure_new_1}
	&v(y; t):=\\
	&\max_{\alpha\in [\ell, y]}\left[\Gamma\cdot\left(\int_{\ell}^{\alpha}\frac{s(\alpha)-s(z)}{s(\alpha)-s(\ell)}\p_y(X_t\in \diff z)+\int_0^t\p_y\left(T_{\ell}\in\diff z\right)\right)+(1-\Gamma)\cdot \dfrac{(s(r)-s(y))s'(\alpha)}{(s(r)-s(\alpha))W_{q,1}(\alpha,\ell)+s'(\alpha)W_q(\alpha,\ell)}\right]\nonumber
	\end{align}
\end{small}
for $s(r)<\infty$, and
\begin{small}
	\begin{equation}\label{eq:measure_new_2}
	v(y; t):=\max_{\alpha\in [\ell, y]}\left[\Gamma\cdot\left(\int_{\ell}^{\alpha}\frac{s(\alpha)-s(z)}{s(\alpha)-s(\ell)}\p_y(X_t\in \diff z)+\int_0^t\p_y\left(T_{\ell}\in\diff z\right)\right)+(1-\Gamma)\cdot \dfrac{s'(\alpha)}{W_{q,1}(\alpha,\ell)}\right]
	\end{equation}
\end{small}
for $s(r)=\infty$.

\begin{remark}\normalfont
	Note that our formulation here is just an example of how $\alpha$ can be chosen for credit risk management. It is up to the party of interest (management) to decide which quantities to use for representing tradeoffs.  Other quantities of interest  include the following: we can consider $\int_0^t \p_y(\lambda_\alpha\in \diff s)$ for a fixed $t$ and $y>\alpha$ from Section \ref{subsec:last-passage}. The quantity $\p_y(\lambda_\alpha\in\diff t)$ decreases as $\alpha$ decreases (i.e., when $\alpha$ approaches $\ell$). This is easily checked if we consider $T_{\alpha}^{\tilde{X}}$ for the reversed process $\tilde{X}$. Furthermore, we can also use the time left until death after the last passage time to $\alpha$, which is equal to the first passage time to $\alpha$ for the $k_r$-transform of $X$ (or $X^*$) (see Propositions \ref{prop:reversal} and \ref{prop:additional}). The company management would want to make this quantity longer, since this would give them some time to recover.
\end{remark}

When $X_t$ is a Brownian motion with drift $\mu$ and unit variance on $(c,\infty)$ with $c$ being a regular killing boundary, \eqref{eq:measure_new_1} and \eqref{eq:measure_new_2} become
\begin{normalsize}
	\begin{equation}\label{eq:objective_positive}
	\begin{aligned}
	v(y; t):= \max_{\alpha\in [c, y]}&\bigg[\Gamma\bigintsss_c^{\alpha}\frac{e^{-2\mu z}-e^{-2\mu\alpha}}{e^{-2\mu c}-e^{-2\mu\alpha}}\frac{e^{-\mu(y-z)-\frac{\mu^2t}{2}}}{\sqrt{2\pi t}}\left(e^{-\frac{(y-z)^2}{2t}}-e^{-\frac{(y+z-2c)^2}{2t}}\right)\diff z+\int_0^t\frac{|c-y|}{\sqrt{2\pi u^3}}e^{-\frac{(c-y-\mu u)^2}{2u}}\diff u\\ &+(1-\Gamma)\frac{e^{-2\mu y}}{e^{-\mu(\alpha+c)}\left(\cosh(\sqrt{\mu^2+2q}(\alpha-c))+\frac{\mu}{\sqrt{\mu^2+2q}}\sinh(\sqrt{\mu^2+2q}(\alpha-c))\right)}\bigg], \quad \mu>0
	\end{aligned}
	\end{equation}
\end{normalsize}
and
\begin{normalsize}
	\begin{equation}\label{eq:objective_negative}
	\begin{aligned}
	v(y; t):= \max_{\alpha\in [c, y]}&\bigg[\Gamma\bigintsss_c^{\alpha}\frac{e^{-2\mu z}-e^{-2\mu\alpha}}{e^{-2\mu c}-e^{-2\mu\alpha}}\frac{e^{-\mu(y-z)-\frac{\mu^2t}{2}}}{\sqrt{2\pi t}}\left(e^{-\frac{(y-z)^2}{2t}}-e^{-\frac{(y+z-2c)^2}{2t}}\right)\diff z+\int_0^t\frac{|c-y|}{\sqrt{2\pi u^3}}e^{-\frac{(c-y-\mu u)^2}{2u}}\diff u\\ &+(1-\Gamma)\frac{e^{-\mu(\alpha-c)}}{\cosh(\sqrt{\mu^2+2q}(\alpha-c))-\frac{\mu}{\sqrt{\mu^2+2q}}\sinh(\sqrt{\mu^2+2q}(\alpha-c))}\bigg], \qquad \mu<0,
	\end{aligned}
	\end{equation}
\end{normalsize}respectively. We used $p(t;\cdot, \cdot)$ in \eqref{eq:original-transition-density} and the hitting time density (\citet[Chapter 3, Sec. 3.5.C]{karatzas}). For the discount rate $q$, one could use the expected return on the company's asset or the weighted average of cost of capital (WACC). In our examples in Section \ref{sec:application}, we defined $q$ as WACC (refer to Appendix \ref{sec:q} for the calculation method).

\section{Application to the Leverage Process}\label{sec:application}

\subsection{Last Passage Time}\label{subsec:last-passage-application}
Below, we illustrate how the last passage time can be useful for risk management. For the analysis, we choose American Apparel Inc., which filed for bankruptcy protection in October 2015\footnote{\url{https://www.theguardian.com/business/2015/oct/05/american-apparel-files-for-bankruptcy} Accessed on 2017/05/27}. The method to choose an appropriate alarming threshold is discussed in Section \ref{sec:endogenizing}. Here we use two alarming thresholds for the analysis: $R^*=\frac{1}{0.8}=1.25$ and $R^*=\frac{1}{0.6}\approx1.67$. These are the levels of the leverage ratio when debt makes up $80\%$ and $60\%$ of assets, respectively.

As discussed in Section \ref{subsec:2-2}, the study  of the leverage process can be reduced to the study of a Brownian motion with drift; therefore, the last passage time of the leverage ratio to $R^*$ is equivalent to the last passage time of the Brownian motion with drift $\mu$ to an appropriate $\alpha$. While $\lambda_{\alpha}$ is not a stopping time, the mathematical foundation provided in our paper enables us to calculate certain probabilities associated with $\lambda_{\alpha}$. Note that all of these probabilities are calculated based on the information of the current position of the leverage process and the information available up to current time. As we will show below, the probabilities associated with $\lambda_{\alpha}$ provide very useful information for risk management.
\subsubsection{Methodology}\label{subsubsec:method}
In order to analyze the distribution of the last passage time to $R^*$, we estimate the necessary  parameters $\nu$ and $\sigma$ of the leverage process by the method in \citet{duan1994}, \citet{duan2000}, and \citet{lehar2005}. This is the estimation method used for \emph{structural} credit risk models and it assumes that the company equity is a European call option written on company assets with a strike price equal to a certain level of debt (see Section \ref{subsec:2-2}). Let us take an example of December 2013 and demonstrate the estimation procedure that we use for each month in the first column of Table \ref{tbl:american-apparel}.
\begin{enumerate}[leftmargin=*]
	\item At the end of December 2013, the estimated drift and volatility parameters ($\nu$ and $\sigma$) of the company's asset process $A$ are $-0.5080$ and $0.2974$, respectively. These parameters were calculated by using the equity and debt data of the previous 6 months. We set debt level $D$ as a sum of ``Revolving credit facilities and current'', ``Cash overdraft", ``Current portion of long-term debt", ``Subord. notes payable - Related Party", and one half of ``Total Long Term Debt", all taken from the company's balance sheet (December 2013). We use daily data for our estimation. Using the quarterly balance sheets, we obtain  daily debt values by interpolation.
	We use 1-year treasury yield at the end of December 2013 as the risk-free rate $r$. We calculate the drift in $\eqref{eq:dynamics-of-X}$ as $\mu=\frac{\nu-r}{\sigma}$.
	\item  We compute the initial asset value $A_0$ (i.e., the value at the end of December 2013) from the Black-Scholes formula (see \citet{lehar2005}) by using $\sigma(=0.2974)$, the equity value $E_0$ and debt level $D_0$ at the end of December 2013:
	\[
	E_0=A_0 \Phi(d_0)-D_0 \Phi(d_0-\sigma\sqrt{T}), \quad \text{where}\quad d_0=\frac{\ln(A_0/D_0)+\frac{\sigma^2 T}{2}}{\sigma\sqrt{T}}.
	\] We set  $T=1$ following \citet{lehar2005}. $\Phi$ is the standard normal distribution function.
	\item We set $B_0^A=0$ and compute $R_0$ as the ratio of $A_0$ and $D_0$ to obtain $R_0=1.8596$. We compute $y$, $\alpha$, and $c$ from $\frac{1}{\sigma}\ln\left(\frac{R D_0}{A_0}\right)$ by setting $R=R_0$, $R=R^*$, and $R=1$, respectively. Note that $y$ is equal to $0$.
\end{enumerate}
\bigskip

For each reference month in Table \ref{tbl:american-apparel}, we can calculate the necessary parameters using the method described above. Note that $R_0$ in Table \ref{tbl:american-apparel} is the value of $R$ at the end of each estimation period, which is equal to  the starting value in an empirical analysis regarding $\lambda_{\alpha}$. We are interested in the following 4 topics:
\begin{itemize}[leftmargin=*]
	\item Starting at $y$, what is the probability that the last passage time will occur within 1 year when the insolvency occurs in finite time? This is expressed by the quantity $\int_0^1\p_y(\lambda_{\alpha}\in \diff t,T_c<\infty)$. The event $T_c<\infty$ makes a difference only in the case of $\mu>0$, since $T_c<\infty$ a.s. for $\mu<0$. Since the parameters change depending on the reference month (Table \ref{tbl:american-apparel}), this quantity changes according to the starting reference point as well. Naturally, higher probability indicates greater credit risk. Note that when $\alpha>y$, the process may never reach
	$\alpha$. For such situations, we have calculated $P_y(\lambda_{\alpha}=0)$ as well.
	\item At any point in time, by varying $R^*$ and studying the last passage time to each level, one can obtain detailed information about credit risk. We set one reference month as a starting point and by varying $R^*$ from $1.2$ to $2.5$, we calculate $\p_y(\lambda_{\alpha}\in[0,1])$ for each corresponding $\alpha$. We emphasize that by varying $R^*$, only the corresponding $\alpha$'s change.
	\item We analyze the relationship between  $\lambda_{\alpha}$ and the default probability (DP). At each reference point, we calculate the probability of defaulting at the end of next year in \citet{merton1974}'s model.  For this, we simulate the asset path for 1 year using the estimated parameters and compare the debt ($D_0e^{r}$) and asset values at the end of 1 year. If the asset value is smaller than debt, we consider it as default. We also calculate $\p_y(T_c<1)$ which is naturally higher than DP, since the latter one only considers the final asset and debt values.
	\item Finally, we calculate the quantity $\p_y(Q_t=\lambda_{\alpha},X_t\in(c,\alpha))$ with $Q_t=\sup\{s<t:X_s=\alpha\}$. This quantity indicates that at a fixed time $t$, if the leverage ratio is below $R^*$, what is the probability that it will never reach $R^*$ again. We use this quantity for our optimization problem in Section \ref{sec:endogenizing} as well. See \eqref{eq:first_quantity} and Remark \ref{rmrk:projection}. For our analysis, we set $t=\frac{1}{4}$ and $t=\frac{1}{2}$, meaning that we are interested in the position after 3 and 6 months. We reproduce \eqref{eq:first_quantity} here
	\[\p_y(Q_t=\lambda_{\alpha},X_t\in(c,\alpha))=\int_{c}^{\alpha}\frac{s(\alpha)-s(z)}{s(\alpha)-s(c)}\p_y(X_t\in\diff z).\]
	Note that this probability may be small due to the fact that $\p_y(X_t\in(c,\alpha))$ is small. This in turn may be the result of $\p_y(T_c<t)$ being high. Therefore, we have calculated all three quantities together for the sake of comparison.
\end{itemize}
\begin{table}[]
	\caption{American Apparel Inc. Parameters (up to 4 decimal points). Standard errors are displayed in parentheses. $A_0$, $D_0$, $R_0$ are the values at the end of each indicated month. 1-year treasury yield curve rate observed at the end of each indicated month was used as $r$. Source: U.S. Department of the Treasury.}
	\label{tbl:american-apparel}
	\resizebox{0.92\textwidth}{!}{\begin{minipage}{\textwidth}
			\begin{tabular}{cccccccccc}
				& $\nu$   & $\sigma$ & $\nu$ s.e. & $\sigma$ s.e. & $\mu$ &  $A_0$  & $D_0$  & $R_0$  & r \\
				\hline\\
				12-Jun & 0.5249  & 0.3940    & (0.1216)     & (0.0323)        & 1.3268  & 214149555 & 125950000 & 1.7003 & 0.0021 \\
				12-Sep & 1.0347  & 0.3127   & (0.2348)     & (0.0435)        & 3.3030 & 290116712 & 126700000 & 2.2898 & 0.0017 \\
				12-Dec & 0.1144  & 0.3720    & (0.0461)     & (0.0158)        & 0.3033 & 223568448 & 117050000 & 1.9100   & 0.0016 \\
				13-Mar & 0.6710   & 0.5376   & (0.1474)     & (0.0488)        & 1.2456  & 362153886 & 129900000 & 2.7879 & 0.0014\\
				13-Jun & 1.3344  & 0.5069   & (0.5055)     & (0.0776)        & 2.6296 & 349193110 & 144100000 & 2.4233 & 0.0015\\
				13-Sep & -0.6840  & 0.3325   & (0.1383)     & (0.0200)        & -2.0604  & 284917506 & 141900000 & 2.0079 & 0.0010 \\
				13-Dec & -0.5080  & 0.2974   & (0.1714)     & (0.0477)        & -1.7128 &  292977497 & 157550000 & 1.8596 & 0.0013 \\
				14-Mar & -0.8271 & 0.7350    & (0.1153)     & (0.0225)        & -1.1270 & 203302448 & 139750000 & 1.4548 & 0.0013 \\
				14-Jun & 0.1555  & 1.0470    & (0.0427)     & (0.0326)        & 0.1475  & 265400126 & 140600000 & 1.8876 & 0.0011 \\
				14-Sep & 0.7723  & 0.6772   & (0.0693)     & (0.0124)        & 1.1384  & 272059468 & 140350000 & 1.9384 & 0.0013 \\
				14-Dec & 0.3089  & 0.6440    & (0.0753)     & (0.1100)        & 0.4757  & 321884369 & 149700000 & 2.1502 & 0.0025 \\
				15-Mar & -0.0881 & 0.5476   & (0.0466)     & (0.0117)        & -0.1657 & 269423743 & 154700000 & 1.7416 & 0.0026 \\
				15-Jun & -0.8783 & 0.3197   & (0.0839)     & (0.0733)        & -2.7562  & 243535792 & 157850000 & 1.5428 & 0.0028
			\end{tabular}
			%\hfill\
	\end{minipage}}
\end{table}

\subsubsection{Evaluations}\label{subsubsec:evaluation}
We discuss the bullet points raised in Section \ref{subsubsec:method} and demonstrate that we can extract more detailed information of credit conditions than when we only know $\p_y(T_c<t)$.  Figure \ref{fig:american-apparel} displays the probability $\int_0^1\p_y(\lambda_{\alpha}\in \diff t,T_c<\infty)$ for two thresholds $R^*=1.25$ and $R^*=1.67$. This probability was calculated by applying Corollary \ref{coro:1} to the process in \eqref{eq:dynamics-of-X}. We note that there is a sharp rise in the graph in 2013, which is consistent with the fact that American Apparel had problems with a new distribution center in 2013. \footnote{For more details about the company's financial situation before going bankrupt, see \url{http://www.wsj.com/articles/american-apparel-ceo-made-crisis-a-pattern-1403742953} Accessed on 2018/11/26} Although the company recovered to some extent during the period of June--December 2014, it went bankrupt in October 2015. The calculation results are summarized in Table \ref{tbl:american-apparel-2} and we discuss them here.

\begin{enumerate}[leftmargin=*]
	\item The two quantities $\int_0^1\p_y(\lambda_{\alpha}\in \diff t,T_c<\infty)$ and $\p_y(\lambda_{\alpha}=0)$ provide additional information to $\p_y(T_c<1)$ and default probability (DP) regarding the creditworthiness of the company. For example, $\p_y(T_c<1)$ and DP are both high in Sep-13. Even though they decrease by more than $10\%$ in Dec-13, $\int_0^1\p_y(\lambda_{\alpha}\in \diff t,T_c<\infty)$ still remains above $80\%$ for $R^*=1.67$. For the leverage ratio, this means that the probability of passing the level $1.67$ last time within 1 year is more than $80\%$ and indicates high credit risk.
	\item Moving to Mar-14, we see that $\p_y(\lambda_{\alpha}=0)$ has increased for $R^*=1.67$. Note that $R_0=1.4548<1.67$ for this reference month (see Table \ref{tbl:american-apparel}); therefore, there is around $44\%$ probability (see Table \ref{tbl:american-apparel-2}) that the firm will become insolvent before the leverage ratio recovers to $1.67$.   We also have a $90\%$ probability of passing the level $1.25$ for the last time within 1 year. The drift $\mu$ turns positive for the next 3 reference points but later becomes negative again.
	\item $\p_y(T_c<1)$ and DP greatly increase from Dec-14 to Mar-15 but the increase in $\int_0^1\p_y(\lambda_{\alpha}\in \diff t,T_c<\infty)$ is even greater for $R^*=1.67$ and signals increased risk. Finally, $R_0=1.5428<1.67$ in Jun-15 and note that $\int_0^1\p_y(\lambda_{\alpha}\in \diff t,T_c<\infty)$ became relatively small (see the red line in Figure \ref{fig:american-apparel}).  However, this is because the probability of never recovering to $1.67$ becomes $74\%$.  This, together with all the other values, indicates an extremely high risk of insolvency. Note also that $\int_0^1\p_y(\lambda_{\alpha}\in \diff t,T_c<\infty)$ is even higher in Jun-15 for the level $R^*=1.25$ (black dotted line), a contrast to $R^*=1.67$.  It follows that by considering several levels of $R^*$, we can see the company's credit conditions more closely. Next, let us further investigate this point.
\end{enumerate}

\begin{figure}[h]
	\includegraphics[scale=0.8]{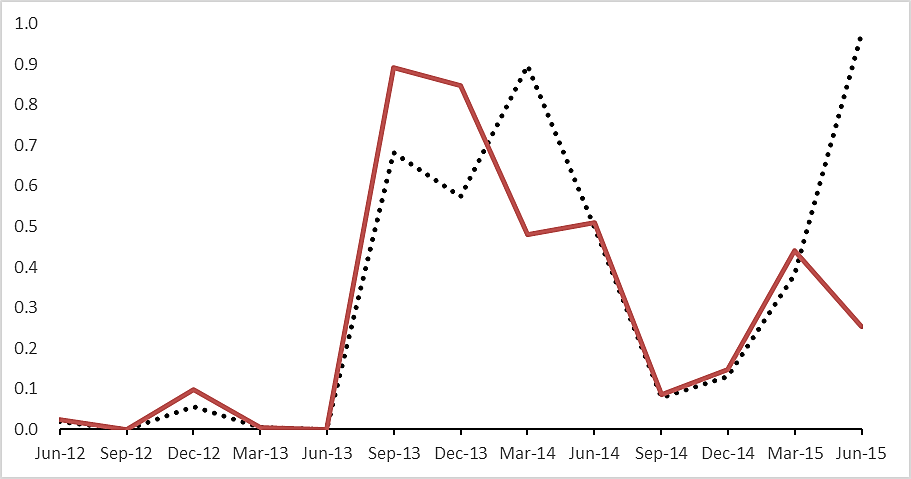}
	\caption{Last Passage Time to $R^*=1.25$ (black dotted line) and $R^*=1.67$ (red line) for American Apparel Inc. For suitable $\alpha$, the graph displays $\int_0^1\p_y(\lambda_{\alpha}\in \diff t, T_c<\infty)$ for Brownian motion with drift. In our setting, $y=0$. The estimated values are available in Table \ref{tbl:american-apparel-2}.}
	\label{fig:american-apparel}
\end{figure}

\begin{table}[h]
	\caption{Calculation results (up to 4 decimal points) based on the parameters in Table \ref{tbl:american-apparel}. In our setting, $y=0$. Note that by changing the level of $R^*$, only corresponding $\alpha$'s change. For each starting point, DP denotes the probability of defaulting at the end of next year.}
	\label{tbl:american-apparel-2}
	\resizebox{0.7\textwidth}{!}{\begin{minipage}{\textwidth}
			\hspace{-3cm}\begin{tabular}{cccccccccccccc}
				&     &    &  & \multicolumn{3}{c}{$R^ *=1.25$}                                               &  & \multicolumn{3}{c}{$R^*=1.67$}                                                                &  &               &        \\
				\cline{5-7} \cline{9-11}
				&$\mu$ & $c$       &  & $\alpha$      & $\int_0^1\p_0(\lambda_{\alpha}\in \diff t,T_c<\infty)$ & $\p_y(\lambda_{\alpha}=0)$ &  & $\alpha$  & $\int_0^1\p_0(\lambda_{\alpha}\in \diff t,T_c<\infty)$ & $\p_y(\lambda_{\alpha}=0)$ &  & $\p_y(T_c<1)$ & DP     \\
				\cline{5-7} \cline{9-11}
				Jun-12 & 1.3268 &-1.3471 &  & -0.7808 & 0.0190                                                 & 0.8740                     &  & -0.0456 & 0.0232                                                 & 0.1140                     &  & 0.0175        & 0.0038 \\
				Sep-12 & 3.3030 &-2.6490 &  & -1.9355 & 0.0000                                                 & 1.0000                     &  & -1.0092 & 0.0000                                                 & 0.9987                     &  & 0.0000        & 0.0000 \\
				Dec-12 & 0.3033 &-1.7397 &  & -1.1398 & 0.0571                                                 & 0.4992                     &  & -0.3610 & 0.0992                                                 & 0.1967                     &  & 0.0468        & 0.0208 \\
				Mar-13 & 1.2456 &-1.9072 &  & -1.4922 & 0.0033                                                 & 0.9757                     &  & -0.9533 & 0.0044                                                 & 0.9070                     &  & 0.0030        & 0.0008 \\
				Jun-13 & 2.6296 &-1.7462 &  & -1.3059 & 0.0001                                                 & 0.9990                     &  & -0.7345 & 0.0001                                                 & 0.9790                     &  & 0.0001        & 0.0000 \\
				Sep-13 & -2.0604 &-2.0966 &  & -1.4255 & 0.6817                                                 & 0                     &  & -0.5542 & 0.8916                                                 & 0                     &  & 0.5767        & 0.4868 \\
				Dec-13 & -1.7128 &-2.0862 &  & -1.3358 & 0.5725                                                 & 0                    &  & -0.3616 & 0.8483                                                 & 0                     &  & 0.4467        & 0.3558 \\
				Mar-14 & -1.1270 &-0.5100 &  & -0.2064 & 0.8968                                                 & 0                     &  & 0.1877  & 0.4803                                                 & 0.4353                     &  & 0.8918        & 0.7293 \\
				Jun-14 & 0.1475 &-0.6068 &  & -0.3937 & 0.4981                                                 & 0.1096                     &  & -0.1170 & 0.5095                                                 & 0.0339                     &  & 0.4954        & 0.2248 \\
				Sep-14 & 1.1384 &-0.9773 &  & -0.6478 & 0.0796                                                 & 0.7712                     &  & -0.2201 & 0.0850                                                 & 0.3941                     &  & 0.0781        & 0.0176 \\
				Dec-14 & 0.4757 &-1.1887 &  & -0.8422 & 0.1295                                                 & 0.5513                     &  & -0.3924 & 0.1483                                                 & 0.3116                     &  & 0.1248        & 0.0484 \\
				Mar-15 & -0.1657 &-1.0131 &  & -0.6056 & 0.3805                                                 & 0                     &  & -0.0766 & 0.4408                                                 & 0                     &  & 0.3652        & 0.1982 \\
				Jun-15 & -2.7562 &-1.3564 &  & -0.6584 & 0.9760                                                 & 0                     &  & 0.2478  & 0.2523                                                 & 0.7449                     &  & 0.9538        & 0.9194
			\end{tabular}
	\end{minipage}}
\end{table}

\indent Let us emphasize that for each fixed point in time, by varying the level $\alpha$, one can obtain more detailed information about credit conditions.
Reviewing these numbers, the management may finetune the company's strategy, investment, and operations. The last passage time not only provides additional information to default probability but also is important in its own right. We illustrate this by looking into Table \ref{tbl:american-apparel-alpha} and Figure \ref{fig:last-passage-time-alpha}.
\newline\indent Suppose that we are at the end of December 2013. The initial position of the leverage ratio is $R_0=1.8596$ for this point in time. We set various levels of $R^*$ (and hence $\alpha$) and compute $\int_0^1\p_y(\lambda_{\alpha}\in\diff t,\lambda_{\alpha}>0)$, the probability that the premonition (i.e., the last passage to $R^*$) occurs within 1 year and $\p_y(\lambda_\alpha=0)$, the probability that the leverage ratio never hits $R^*$. Note that for this period, the drift of Brownian motion is $\mu=-1.7128$, and therefore, $T_c<\infty$ a.s..
\begin{itemize}[leftmargin=*]
	\item [(4)] We see that there is more than $75\%$ probability that the leverage ratio will never recover to $2.1$ or higher levels. Even though the current level is $R_0=1.8596$, this indicates that the levels above $2.1$ are too high in credit quality relative to the company's current position.
	\item [(5)] Next, take $R^*=1.9$. Then, $\p_y(\lambda_\alpha=0)$ is down to $0.2195$, while the probability that the last passage to $R^*=1.9$ occurs within 1 year is $0.7045$. These values together indicate that the probability of never recovering even to the level $1.9$ after 1 year is more than $92\%$. This is valuable information because the decrease in default probability from Sep-13 to Dec-13 and increase in $\mu$ (Table \ref{tbl:american-apparel-2}) may be interpreted as credit quality improvement, while the information in Table \ref{tbl:american-apparel-alpha} clearly indicates the presence of high risk and that the prognosis is not good.
	\item [(6)] If we go further down to $R^*=1.5$, we have more than $75\%$ chance of passing the levels $1.5$ and above for the last time within 1 year. Even for $R^*=1.2$, we have  $\int_0^1\p_y(\lambda_{\alpha}\in\diff t,\lambda_{\alpha}>0)=0.5347$. This means that there is about 50\% chance that the company passes $R^*=1.2$ within 1 year and will not return to this level based on the current position of $R_0=1.8596$.
\end{itemize}
The quantities associated with the last passage time are more stable than default probabilities. See points (1) and (5) above. These quantities should prevent the management from becoming too optimistic when the risk still persists.
\begin{table}[]
	\caption{American Apparel Inc. last passage time probabilities (up to 4 decimal points) for varying $R^*$ by using the end of December 2013 as the starting point (see Table \ref{tbl:american-apparel}). We note that the change of $R^*$ affects only $\alpha$ and all other parameters are unchanged. The dotted line between $R^*=1.9$ and $R^*=1.8$ indicates that $R_0$ in December 2013 is located between the two numbers. Note that $y=0$ in our setting.}
	\label{tbl:american-apparel-alpha}
	\resizebox{0.8\textwidth}{!}{\begin{minipage}{\textwidth}
			\hspace{1cm}\begin{tabular}{cccccccccccc}
				$R^*$   & $\alpha$ & $\int_0^1\p_0(\lambda_{\alpha}\in\diff t,\lambda_{\alpha}>0)$
				& $\p_0(\lambda_{\alpha}=0)$
				& $\p_0(\lambda_{\alpha}\in[0,1])$
				&  &  &  &  &  &  &  \\
				\hline\\
				2.5 & 0.9952   & 0.0219 & 0.9670 & 0.9888 &  &  &  &  &  &  &  \\
				2.4 & 0.8579   & 0.0375 & 0.9471 & 0.9846 &  &  &  &  &  &  &  \\
				2.3 & 0.7148   & 0.0651 & 0.9137 & 0.9787 &  &  &  &  &  &  &  \\
				2.2 & 0.5653   & 0.1147 & 0.8559 & 0.9706 &  &  &  &  &  &  &  \\
				2.1 & 0.4089   & 0.2058 & 0.7537 & 0.9596 &  &  &  &  &  &  &  \\
				2   & 0.2448   & 0.3766 & 0.5679 & 0.9445 &  &  &  &  &  &  &  \\
				1.9 & 0.0723   & 0.7045 & 0.2195 & 0.9241 &  &  &  &  &  &  &  \\
				\dotfill & \dotfill &\dotfill &\dotfill &\dotfill &\dotfill &\\
				1.8 & -0.1095  & 0.8968 & 0      & 0.8968 &  &  &  &  &  &  &  \\
				1.7 & -0.3017  & 0.8610 & 0      & 0.8610 &  &  &  &  &  &  &  \\
				1.6 & -0.5056  & 0.8149 & 0      & 0.8149 &  &  &  &  &  &  &  \\
				1.5 & -0.7227  & 0.7574 & 0      & 0.7574 &  &  &  &  &  &  &  \\
				1.4 & -0.9547  & 0.6887 & 0      & 0.6887 &  &  &  &  &  &  &  \\
				1.3 & -1.2039  & 0.6118 & 0      & 0.6118 &  &  &  &  &  &  &  \\
				1.2 & -1.4731  & 0.5347 & 0      & 0.5347 &  &  &  &  &  &  &
			\end{tabular}
	\end{minipage}}
\end{table}
\begin{figure}[h]
	\includegraphics[scale=0.45]{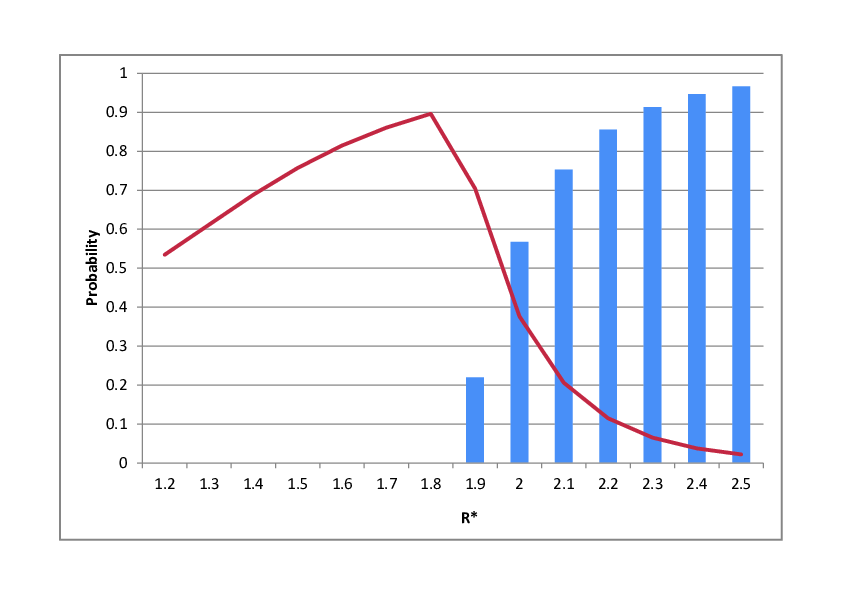}
	\caption{Last passage time to varying $R^*$ for American Apparel Inc. by using the end of December 2013 as the starting point (see Tables \ref{tbl:american-apparel} and \ref{tbl:american-apparel-alpha}). $R_0=1.8596$. The horizontal axis displays $R^*$ and the vertical axis displays $\int_0^1\p_y(\lambda_{\alpha}\in \diff t, \lambda_{\alpha}>0)$ (red lines) and $\p_y(\lambda_\alpha=0)$ (blue bars) for Brownian motion with drift for suitable $\alpha$. Note that $y=0$ in our setting.}\label{last-passage-time-alpha}
	\label{fig:last-passage-time-alpha}
\end{figure}
\newline \indent Finally, we calculate  $\p_y(Q_t=\lambda_{\alpha},X_t\in(c,\alpha))$ with $Q_t=\sup\{s<t:X_s=\alpha\}$. This quantity indicates that at a fixed time $t$, if the leverage ratio is below $R^*$, what is the probability that it will never reach $R^*$ again.  For our analysis, we set $t=\frac{1}{4}$ and $t=\frac{1}{2}$. The results are displayed in Table \ref{tbl:american-apparel-3} where we have chosen $\alpha$ corresponding to $R^*=1.67$.
\begin{itemize}[leftmargin=*]
	\item [(7)] At the end of September 2013, the probability of insolvency happening in 3 and 6 months are 0.13\% and 9.34\%, respectively. However, the probability that the leverage ratio never recovers to the level of 1.67 if it is below 1.67 after 3 (6) months is 30.88\% (53.44\%) and signals high risk.
	\item [(8)] The probability of becoming insolvent $\p_y(T_c<t)$ decreases in December 2013 for both $t=\frac{1}{4}$ and $t=\frac{1}{2}$. To the contrary, we see an increase in $\p_y(Q_t=\lambda_{\alpha},X_t\in(c,\alpha))$ which in turn indicates the still remaining high credit risk and provides valuable information to the management.
	\item [(9)] At the end of March 2014, $\p_y(T_c<t)$ becomes more than 50\% (70\%) for $t=\frac{1}{4}$ ($t=\frac{1}{2}$); therefore, $\p_y(X_t\in(c,\alpha))$ becomes smaller than the numbers in the previous data point Dec-13, which in turn decreases $\p_y(Q_t=\lambda_{\alpha},X_t\in(c,\alpha))$. Note that $R_0=1.4548<1.67$ for this reference point and the probability of never recovering to 1.67 is 43.53\%,  a  hike from the previous data point (see Table \ref{tbl:american-apparel-2}).
	\item [(10)]  The situation starts to worsen after December 2014. In June 2015, $R_0=1.5428<1.67$ and there is roughly 13\% probability of becoming insolvent within next 3 months (see the last row in Table \ref{tbl:american-apparel-3} with $t=1/4$). In addition to this, we find that there is 84\% chance that the leverage ratio will be below 1.67 after 3 months and when it is the case, it will never recover back to 1.67 with the probability of 80\%. For $t=\frac{1}{2}$,  $\p_y(T_c<t)$ is greater than 60\%; therefore, $\p_y(X_t\in(c,\alpha))$ and $\p_y(Q_t=\lambda_{\alpha},X_t\in(c,\alpha))$ are less than 40\%. Moreover, we know that there is more than 74\% chance that the leverage ratio never recovers to the level of 1.67 (see Table \ref{tbl:american-apparel-2}).
\end{itemize}

\begin{table}[h]
	\centering
	\caption{Summary of the results for $R^*=1.67$. In our setting, $y=0$.}
	\label{tbl:american-apparel-3}
	\resizebox{0.7\textwidth}{!}{\begin{minipage}{\textwidth}
			\hspace{-3cm}\begin{tabular}{ccccccccccccc}
				&          &         &         & \multicolumn{3}{c}{$t=\frac{1}{4}$}                                                       &  & \multicolumn{3}{c}{$t=\frac{1}{2}$}                                                     &  &  \\
				\cline{5-7}\cline{9-11}
				& $\alpha$ & $c$     & $\mu$   & $\p_y(Q_t=\lambda_{\alpha},X_t\in(c,\alpha))$ & $\p_y(X_t\in (c,\alpha))$ & $\p_y(T_c<t)$ &  & $\p_y(Q_t=\lambda_{\alpha},X_t\in(c,\alpha))$ & $\p_y(X_t\in(c,\alpha))$ & $\p_y(T_c<t)$ &  &  \\
				\cline{5-7}\cline{9-11}
				Jun-12 & -0.0456  & -1.3471 & 1.3268  & 0.0124                                        & 0.2243                    & 0.0010        &  & 0.0113                                        & 0.1512                   & 0.0069       &  &  \\
				Sep-12 & -1.0092  & -2.6490 & 3.3030  & 0.0000                                        & 0.0001                    & 0.0000        &  & 0.0000                                        & 0.0001                   & 0.0000       &  &  \\
				Dec-12 & -0.3610  & -1.7397 & 0.3033  & 0.0281                                        & 0.1908                    & 0.0003        &  & 0.0485                                        & 0.2262                   & 0.0080       &  &  \\
				Mar-13 & -0.9533  & -1.9072 & 1.2456  & 0.0004                                        & 0.0057                    & 0.0000        &  & 0.0013                                        & 0.0124                   & 0.0005       &  &  \\
				Jun-13 & -0.7345  & -1.7462 & 2.6296  & 0.0000                                        & 0.0027                    & 0.0000        &  & 0.0000                                        & 0.0018                   & 0.0000       &  &  \\
				Sep-13 & -0.5542  & -2.0966 & -2.0604 & 0.3088                                        & 0.4676                    & 0.0013        &  & 0.5344                                        & 0.6562                   & 0.0934       &  &  \\
				Dec-13 & -0.3616  & -2.0862 & -1.7128 & 0.3561                                        & 0.5522                    & 0.0008        &  & 0.5519                                        & 0.6969                   & 0.0611       &  &  \\
				Mar-14 & 0.1877   & -0.5100 & -1.1270 & 0.1853                                        & 0.3278                    & 0.5029        &  & 0.0774                                        & 0.1416                   & 0.7338       &  &  \\
				Jun-14 & -0.1170  & -0.6068 & 0.1475  & 0.0628                                        & 0.1881                    & 0.2053        &  & 0.0326                                        & 0.0991                   & 0.3564       &  &  \\
				Sep-14 & -0.2201  & -0.9773 & 1.1384  & 0.0238                                        & 0.1418                    & 0.0148        &  & 0.0166                                        & 0.0927                   & 0.0448       &  &  \\
				Dec-14 & -0.3924  & -1.1887 & 0.4757  & 0.0326                                        & 0.1436                    & 0.0097        &  & 0.0357                                        & 0.1379                   & 0.0507       &  &  \\
				Mar-15 & -0.0766  & -1.0131 & -0.1657 & 0.1550                                        & 0.4215                    & 0.0504        &  & 0.1283                                        & 0.3274                   & 0.1789       &  &  \\
				Jun-15 & 0.2478   & -1.3564 & -2.7562 & 0.7925                                        & 0.8406                    & 0.1289        &  & 0.3659                                        & 0.3798                   & 0.6095       &  &
			\end{tabular}
	\end{minipage}}
\end{table}
We should keep in mind that the occurrence of the last passage to $R^*$ within a certain period of time does \emph{not} mean that insolvency occurs within that period. To obtain more information in this respect,  we consider,  in the next section, the time interval between the last passage time to a state and the subsequent insolvency.

\subsection{Time Left Until Insolvency After the Last Passage Time}\label{subsec:4-special}
We use Proposition \ref{prop:additional} in order to analyze the density $\p_\cdot(T_c-\lambda_\alpha\in \diff t)$. Figure \ref{Fig-problem2} displays the density $\p_\cdot(T_c-\lambda_\alpha\in \diff t)$ for the reference point December 2013 (see Table \ref{tbl:american-apparel}) when $\alpha=-1.3358$, $c=-2.0862$, and $\mu=\frac{\nu-r}{\sigma}=-1.7128$ (taken from Table \ref{tbl:american-apparel-2}). This is the density of the time left until insolvency after the last passage time to $R^*=1.25$. We used Zakian's method described in \citet{halsted-brown1972} to obtain the density. We want to emphasize that when taking the limit $x\to c$, the density in \eqref{eq:answer-to-problem2} converges. It is seen in Figure \ref{Fig-problem2} that the distribution is dense in the range of $t=0.1\sim 0.2$, so that insolvency is rather imminent after passing $\alpha=-1.3358$, since the company's asset value has the negative drift parameter $\mu=-1.7128$.  From the numerical result in Section \ref{subsubsec:evaluation},  for $R^*=1.25$, the probability $\int_0^1\p_y(\lambda_\alpha\in \diff t, \lambda_\alpha>0)$ is 0.5725 (see Table \ref{tbl:american-apparel-2}).  Based on the analysis here,  if the last passage time occurs, the time left for the management to improve credit quality is only a month or so.
\begin{figure}[h]
	\begin{center}
		\centering{\includegraphics[scale=0.55]{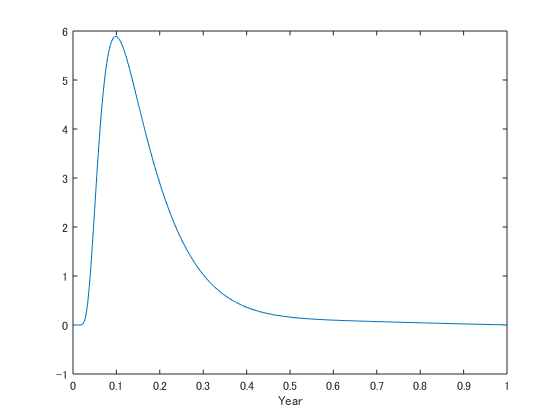}}\\
		\caption{The density of $\p_\cdot(T_c-\lambda_\alpha\in \diff t)$ of American Apparel Inc., where we take $\alpha=-1.3358$ (corresponding to $R^*=1.25$), $c=-2.0862$, and $\mu=\frac{\nu-r}{\sigma}=-1.7128$ of the company's  December 2013 figures from Tables \ref{tbl:american-apparel} and \ref{tbl:american-apparel-2}.} \label{Fig-problem2}
	\end{center}
\end{figure}

\subsection{Endogenizing the Threshold}
In this section, we solve the optimization problem from Section \ref{sec:endogenizing} to endogeneously obtain the threshold $\alpha$ of interest for the leverage ratio. We analyze the results and their implications.
\subsubsection{The effect of $\Gamma$ on the optimal $\alpha$}\label{subsubsec:gamma-effect}
Figure \ref{fig:alpha_gamma} displays the optimal values of $\alpha$ (we call it $\alpha^*$) for each $\Gamma$ by using the end of December 2013 as a reference point. As it is expected, when the second element of the objective function has the priority (i.e., small $\Gamma$), the optimal $\alpha$ is low. This renders $A_{\infty}$ into a small value, making the Laplace transform greater. For $\Gamma\le0.3$, we have a corner solution and $\alpha^*=c$. To the contrary, when the first element has the priority (i.e., large $\Gamma$), optimal $\alpha$ increases, giving a sense of danger as a precaution of potential threat even at high levels of $\alpha$. We remind the reader that $y=0$ in our setting; therefore, for large $\Gamma$ we have a corner solution and $\alpha^*=y$. As a closer look in Figure \ref{fig:alpha_gamma_0.0005_inner} reveals, our optimization problem has an inner solution for $\Gamma\in(0.3,0.5)$.
\begin{figure}[]
	\begin{subfigure}{0.5\textwidth}
		\includegraphics[scale=0.5]{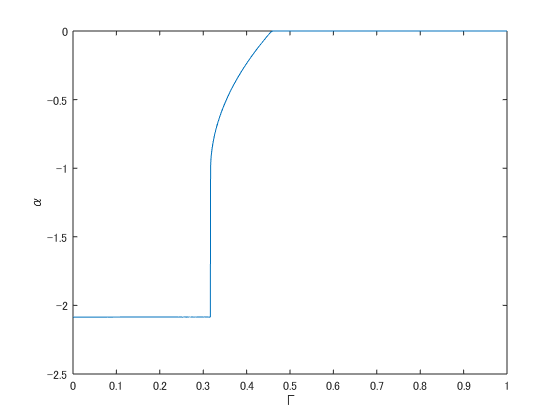}
		\caption{}
		\label{fig:alpha_gamma_0.0005}
	\end{subfigure}\\
	\begin{subfigure}{0.5\textwidth}
		\includegraphics[scale=0.5]{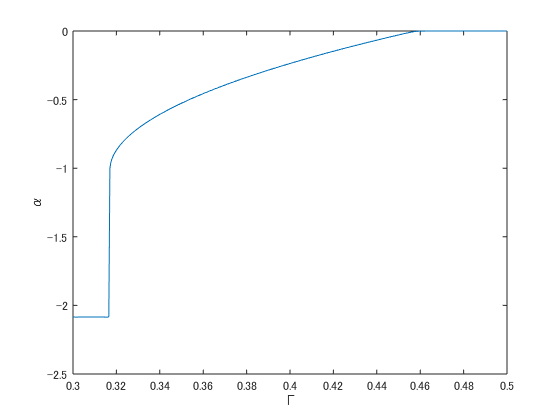}
		\caption{}
		\label{fig:alpha_gamma_0.0005_inner}
	\end{subfigure}
	\caption{Optimal values of $\alpha$ for $\Gamma\in[0,1]$ (a) and a closer look for $\Gamma\in[0.3,5]$ (b). The end of December 2013 (see Table \ref{tbl:american-apparel-2}) is chosen as a starting point. $y=0$, $c=-2.0862$, $\mu=-1.7128$, $q=0.3006$, $t=1$.}
	\label{fig:alpha_gamma}
\end{figure}
\subsubsection{Comparative statics}
For comparative statics, we set $\Gamma=0.4$. Table \ref{tbl:comparative_alpha} in Appendix \ref{sec:q} displays the optimal $\alpha$ (that is, $\alpha^*$) for \eqref{eq:objective_positive} and \eqref{eq:objective_negative} with different values of $\mu$ and $q$. Recall that  we defined $q$ as WACC (see Appendix \ref{sec:q} for the calculation method). The pattern of $\alpha^*$ does not depend on the initial value of $\alpha$ used during optimization; however, some values in the table may slightly change with different initial values. The combination of high $\mu$ and low $q$ indicates less risky situations. The lower-right corner of the table shows the riskiest conditions. Let us take a closer look at the entries where the optimization has inner solutions. For less risky situations, we observe that the time spent in financial distress discussed in (2) is considered high and the management wants to make $A_{\infty}$ as small as possible; hence, $\alpha^*$ is close to $c$. We notice the opposite results for the risky area (lower-right part) of the table. Let us fix $\mu$ and observe how the optimal level of $\alpha$ varies as $q$ changes. If we look at $\mu<-1.2$, we see that for fixed $\mu$, as $q$ increases, $\alpha^*$ gradually decreases. This means that $A_{\infty}$ decreases in order to balance the product $q\cdot A_{\infty}$. The parameter $q$ affects only the second element of the objective function. For $\mu >-1.2$, the situation is of ``bang-bang'' type: rather than smoothly decreasing values of $\alpha^*$, it is suggested that the optimal level of $\alpha$ be jumped from zero to a level close to $c$ as $q$ increases. In less risky situations, one can be rather bold in reducing the level $\alpha^*$ to a point close to $c$. Finally, looking at the lower-right part of the table, for fixed $q$, as $\mu$ decreases, $\alpha^*$ also decreases. Decreasing $\mu$ indicates an increasing risk of insolvency. The first (resp. second) element of the objective function \eqref{eq:objective_negative} would move $\alpha^*$ to higher (resp. lower) value. Since we have set $\Gamma=0.4$, we give more weight to the second element, and thus, $\alpha^*$ decreases.

\subsection{The effect of $\alpha$ on the asset value dynamics}\label{subsubsec:alpha-effect} Next, we consider the influence the decisions of the management have on the dynamics of the asset value once the leverage ratio is below the premonition level $R^*$ (which has a one-to-one correspondence with an appropriate $\alpha$). We continue the discussion with $R^*$ and insolvency level 1, rather than $\alpha$ and $c$ since it will be easier to interpret the results of our analysis. Note that the asset dynamics is expressed by the equation $A_t=A_0e^{\nu t+\sigma B_t^A}$ (see Section \ref{subsec:2-2}). \\
%\underline{Part 1:}\\
\indent We illustrate the effect of the management's decisions by simulation and take two reference points for this: December 2012 and December 2013 (see Table \ref{tbl:american-apparel}). The credit condition of American Apparel Inc. is quite different for these two reference points. One way to see this is to look at WACC. At the end of December 2013, WACC is $30.06\%$ in contrast to $11.84\%$ at the end of December 2012. Furthermore, $\nu-r$ is negative (positive) for December 2013 (2012) where $r$ denotes the risk-free rate. To deal with these two contrasting scenarios, we consider different management strategies for Dec-2012 and Dec-2013. This is reasonable, since the parameter $\nu$ of the asset process has opposite signs for these two reference points. The strategies below change the parameters gradually. Furthermore, we have kept the ratio of the change in $\nu$ to change in $\sigma$ roughly at 1.6 for all strategies. Note that the level $R^*$ is decided first and the strategies described below are implemented with respect to the given $R^*$.\\

\noindent \underline{$\nu-r>0$} : December 2012
\begin{enumerate}[leftmargin=*]
	\item [1.]  The management does not change its decisions regardless of whether the leverage ratio is below $R^*$ or not. That is, $\nu$ and $\sigma$ are unchanged.
	\item [2.] When the leverage ratio is at or below $R^*$, the management acts on behalf of the creditors and replaces the risky investments with less risky ones. For this, we subtract 0.0005 from $\nu$ and 0.0003 from $\sigma$ for the next simulation step only when the firm is operating at or below $R^*$. We allow this change in parameters while $\nu-r\ge0$ and $\sigma>0$.
	\item  [3.] When the leverage ratio is at or below $R^*$, the management acts on behalf of the shareholders and makes riskier investments. For this, we add 0.0008 to $\nu$ and 0.0005 to $\sigma$ for the next simulation step only when the firm is operating at or below $R^*$.
\end{enumerate}

\noindent \underline{$\nu-r<0$}: December 2013
\begin{enumerate}[leftmargin=*]
	\item [1.] The same as the strategy 1. for $\nu-r>0$.
	\item [2.] When the leverage ratio is at or below $R^*$, the management acts on behalf of the creditors and makes new investments with small risk in order to lift the negative drift up. For this, we add 0.0005 to $\nu$ and 0.0003 to $\sigma$ for the next simulation step only when the firm is operating at or below $R^*$.
	\item [3.]  When the leverage ratio is at or below $R^*$, the management acts on behalf of the shareholders and makes riskier investments. To compare the results with the strategy 2. here, we add 0.0015 to $\nu$ and 0.0009 to $\sigma$ for the next simulation step only when the firm is operating at or below $R^*$.
\end{enumerate}
See Table \ref{tbl:sim1} for the results based on 50,000 simulated paths. The time horizon is one year.  For comparison, we have used three levels of $R^*$. The first one corresponds to the solution of the optimization problem in \eqref{eq:objective_positive} and \eqref{eq:objective_negative} with $t=1$ and $\Gamma=0.4$. This is an objective level of $R^*$ that does not take into account the difference in the risk aversion of shareholders and creditors. The second and third threshold levels are $1.67$ and $1.25$, respectively (see Table \ref{tbl:american-apparel-2}).
The parameters ``change in $\nu$" and  ``change in $\sigma$" indicate the size of the change for the next simulation step when the present leverage ratio is at or below $R^*$.  We set the starting point of the leverage ratio as $\frac{A_0}{D_0}$ (see Table \ref{tbl:american-apparel}). For each strategy, we calculated two quantities of interest: the probability of becoming insolvent within 1 year $\p_{\cdot}(\exists t \; 0\le t\le 1; R_t\le 1)$ and the fraction of  1 year spent above $R^*$. 
\newline\indent As for December 2012, since $R^*=1.0022$ is close to $1$, we do not see any significant difference between the strategies. This level of $R^*$ is a result of $\Gamma=0.4$ that gives priority to the second term in the optimization problem. Since $R_0=1.91$ is much greater than $R^*$ and the $\nu-r$ is positive, naturally, almost the whole 1 year is spent above $R^*$ and $A_{\infty}$ is small. For $R^*=1.67$, we see that the strategy ``creditors'' results in smaller insolvency probability. This should be due to the smaller $\sigma$ that leads to less fluctuations and hence less probability of becoming insolvent.  On the other hand, since this strategy makes the drift smaller, the process is less likely to go above $R^*$.  Since $R^*=1.25$ gets closer to 1, the performances become similar to $R^*=1.0022$. As for December 2013, to our surprise, the performance of the riskier strategy is better. A possible explanation is the following:  Since the drift is quite small at the end of December 2013, the strategy that greatly increases the drift becomes necessary. Therefore, the strategy ``shareholders" results in the lowest probability of becoming insolvent. For this strategy, the time spent above $R^*$ is also the largest.

\begin{table}[h]
	\caption{Simulation results for December 2012 and 2013. The values are displayed up to 4 decimal points. For each $R^*$ and strategy, ``Insolvency" indicates the probability of becoming insolvent within 1 year. ``Time above $R^*$" indicates the fraction of 1 year spent above $R^*$ by the leverage ratio.}
	\label{tbl:sim1}
	\resizebox{0.7\textwidth}{!}{\begin{minipage}{\textwidth}
			\hspace{-1cm}\begin{tabular}{ccccccccccc}
				December 2012\\
				\hline
				Strategy & change in $\nu$ & change in  $\sigma$ &  &  &                &     Parameters          &  &  &                &               \\
				\cline{1-3}\cline{7-7}
				1. No change        & 0               & 0                                                              &  &  &                &   $R_0$            &  1.9100 &  &                &               \\
				2. Creditors        & -0.0005         & -0.0003                                                        &  &  &                &   $\nu$            & 0.1144 &  &                &               \\
				3. Shareholders        & 0.0008          & 0.0005                                                         &  &  &                &      $\sigma$         & 0.3720 &  &                &               \\
				&                 &                                                                &  &  &                &               &  &  &                &               \\
				& \multicolumn{2}{c}{$R^*=1.0022$}                                                                  &  &  & \multicolumn{2}{c}{$R^*=1.67$} &  &  & \multicolumn{2}{c}{$R^*=1.25$} \\
				\cline{2-3}\cline{6-7}\cline{10-11}
				Strategy & Insolvency      & Time above $R^*$                                                      &  &  & Insolvency     & Time above $R^*$    &  &  & Insolvency     & Time above $R^*$   \\
				1. No change        & 0.0439         & 0.9872                                                       &  &  & 0.0439        & 0.7939      &  &  & 0.0439        & 0.9654      \\
				2. Creditors        & 0.0432         & 0.9874                                                        &  &  & 0.0380        & 0.7840      &  &  & 0.0425        & 0.9643     \\
				3. Shareholders        & 0.0443         & 0.9867                                                       &  &  & 0.0533         & 0.8061      &  &  & 0.0453        & 0.9657      \\
				& & & & & & & & & & \\
				December 2013\\
				\hline
				Strategy & change in $\nu$ & change in  $\sigma$ &  &  &   &    Parameters                &  &  &                &               \\
				\cline{1-3}\cline{7-7}
				1. No change        & 0               & 0                                                              &  &  &                &    $R_0$           & 1.8596 &  &                &               \\
				2. Creditors        & 0.0005          & 0.0003                                                         &  &  &                &   $\nu$            & -0.5080 &  &                &               \\
				3. Shareholders        & 0.0015          & 0.0009                                                         &  &  &                &      $\sigma$         & 0.2974 &  &                &               \\
				& & & & & & & & & & \\
				& \multicolumn{2}{c}{$R^*=1.7332$}                                                 &  &  & \multicolumn{2}{c}{$R^*=1.67$} &  &  & \multicolumn{2}{c}{$R^*=1.25$} \\
				\cline{2-3}\cline{6-7}\cline{10-11}
				Strategy & Insolvency      & Time above $R^*$                                                  &  &  & Insolvency     & Time above $R^*$    &  &  & Insolvency     & Time above $R^*$   \\
				1. No change        & 0.4343         & 0.2685                                                       &  &  & 0.4343        & 0.3297      &  &  & 0.4343        & 0.7300      \\
				2. Creditors        & 0.4124         & 0.2887                                                       &  &  & 0.4130          & 0.3488      &  &  & 0.4295        & 0.7355      \\
				3. Shareholders        & 0.3849         & 0.3380                                                        &  &  & 0.3899        & 0.3925      &  &  & 0.4213         & 0.7431      \\
			\end{tabular}
	\end{minipage}}
\end{table}

\begin{figure}[H]
	\centering{\includegraphics[scale=0.7]{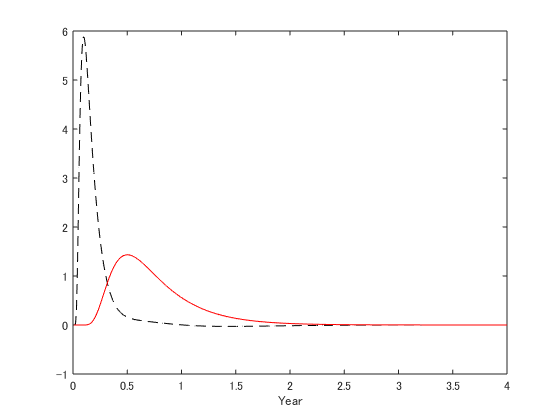}}
	\caption{The density of the distribution $\p_\cdot(T_c-\lambda_{\alpha^*}\in \diff t)$ for American Apparel Inc. for $\alpha^*=-0.2367$ (red line) and $\alpha^*=-1.3358$ (black dashed line). The end of December 2013 (see Tables \ref{tbl:american-apparel} and \ref{tbl:american-apparel-2}) is used as the reference point. $c=-2.0862$, $\mu=\frac{\nu-r}{\sigma}=-1.7128$.} \label{Fig-problem2-2}
\end{figure}
\subsection{Summary of the Risk Management Tool}\label{sec:conclusion}
We summarize the proposed risk management tool in this paper: the management watches the company's leverage ratio and they can estimate the drift and variance parameters of the firm's asset value process, based on the method described in Section \ref{subsubsec:method}. Then, for any future time horizon, the management can determine the threshold level $\alpha^*$ below which the company should be operated on alert and with precaution. For this $\alpha^*$, the management can compute $\p_\cdot(\lambda_{\alpha^*}\in \diff t)$ (together with other associated probabilities in Sections \ref{subsubsec:method} and \ref{subsubsec:evaluation}) and $\p_\cdot(T_c-\lambda_{\alpha^*}\in \diff t)$ to make plans for future business operations. For example, suppose that the company has parameters of December 2013 (see Table \ref{tbl:american-apparel}) and solves \eqref{eq:objective_negative} with $\Gamma=0.4$ and $y=0$ to obtain the solution $\alpha^*=-0.2367$. Then, the red line in Figure \ref{Fig-problem2-2} is the density of the distribution $\p_\cdot(T_c-\lambda_{\alpha^*}\in \diff t)$ for $\alpha^*=-0.2367$ (see Proposition \ref{prop:additional}). In this case, the density is clustered around $0.5$ years. This means that if the management sets $\alpha$ as $-0.2367$, given the company's current leverage level and asset growth rate, it is not unnatural to assume that there will be still half a year between the \emph{last} passage to $\alpha$ and the insolvency. The black dashed line is a replication of the density in Figure \ref{Fig-problem2} where the time between $\lambda_{\alpha^*}$ and insolvency is clustered around $0.1\sim 0.2$ years. This density  was computed under the arbitrary assumption that $\alpha^*=-1.3358$ which corresponds to debt being 80\% of assets.  It follows that one could be better-off  by setting $\alpha^*$ higher based on the optimization of $\alpha$, rather than arbitrarily setting the level, to avoid further deterioration and to have a longer period until insolvency even after the last passage of $\alpha^*$. Hence, together with $\p_y(\lambda_{\alpha^*}\in \diff t)$, one can extract information useful for risk management by the techniques presented in Sections \ref{sec:general} and \ref{sec:application}.

\begin{remark}
	\begin{normalfont}
		In the end, we comment on the possibility of extending our results in Section \ref{sec:general} to L\'{e}vy processes. Our analysis relies on the scale functions, speed measures, and Green functions of diffusions. For spectrally negative L\'{e}vy processes, the scale function (usually denoted by $W(\cdot)$) and the process conditioned to stay positive (which is created by $W(\cdot)$) are well-studied. See  \citet[Chapter 8]{Kyprianou_2014}, \citet{Kuznetsov2061}, and \citet{Bertoin_1996}. As for the reversal of L\'{e}vy processes, it is well-known that the dual of $X$, which is $-X$, has the same law as the reverse of $X$  from a fixed time (\citet[Chapter 2]{Bertoin_1996}). However, to our knowledge, the transform to make the process go to a specific state is not available. To extend the analysis in this article to L\'{e}vy processes, this task may be necessary.
	\end{normalfont}
\end{remark}

\begin{appendix}
\section{}\label{sec:q}
We demonstrate the calculation procedure for $q$ in  \eqref{eq:objective_positive} and \eqref{eq:objective_negative} by using the end of December 2013 as a reference point. The steps for estimating $q$ for the end of December 2013 are as follows:
	\begin{enumerate}[leftmargin=*]
		\item The value of equity $E$ is calculated as a difference between the market value of the asset $A$ (estimated by the method in Section \ref{subsubsec:method}) and the value of debt $D$ (used in the estimation of $A$). Using these estimates, the weights of the equity and debt are determined as $w_E=\frac{E}{A}$ and $w_D=\frac{D}{A}$, respectively.
		\item Using 1-year daily time-series (of 2013) of NASDAQ Composite Price Index and of the share price of the company of interest, we calculate daily returns of these two time-series. Regression (including an intercept) of the company's daily returns on the index's daily returns gives us the company $\beta$. With this $\beta$, the annual NASDAQ return ($R_m$) for 2013, and using the 1-year U.S. Treasury yield curve rate ($R_f$) for the end of 2013 as the risk-free rate (see Table \ref{tbl:american-apparel}), we calculate the annual expected return for the company. This is the cost of equity $C_E$.
		\item Dividing ''Interest paid" (from the cash flow statement) of year 2013 by the average of the debt values of December 2012 and December 2013 (the calculation of debt values is described in Section \ref{subsubsec:method}) gives us the cost of debt $C_D$.
		\item Setting the corporate tax rate as $35\%$, $q=w_E\cdot C_E+w_D\cdot C_D\cdot(1-0.35)$.
	\end{enumerate}
	
	\begin{table}[h]
		\centering
		\caption{Results of The Abovementioned Procedure for American Apparel Inc. for the end of December 2012 and December 2013 . $A$, $D$, $E$, and Interest Paid are displayed in millions of U.S. dollars. $q_{2012}=11.8445\%$ and $q_{2013}=30.0572\%$. The values in the table are rounded up to the 2nd decimal point.}
		\label{tbl:q}
		\resizebox{1\textwidth}{!}{\begin{minipage}{\textwidth}
			\hspace{1cm}	\begin{tabular}{|c|c||c|c||c|c|}
					\hline
					\multicolumn{6}{|c|}{December 2012}\\
					\hline
					$R_f$   & 0.16$\%$  & Interest Paid & 10.95 & $A$   & 223.57    \\ \hline
					$R_m$   & 15.91$\%$ & $D_{2012}$    & 117.05
					& $E$   & 106.52   \\ \hline
					$\beta$ & 1.11      & $D_{2011}$     & 101.75
					& $w_E$ & 47.64$\%$ \\ \hline
					$C_E$   & 17.71$\%$ & $C_D$         & 10.01$\%$ & $w_D$ & 52.36$\%$ \\ \hline
		\end{tabular}\end{minipage}}
		\resizebox{1\textwidth}{!}{\begin{minipage}{\textwidth}\hspace{1.1cm}\begin{tabular}{|c|c||c|c||c|c|}
					\hline
					\multicolumn{6}{|c|}{December 2013}\\
					\hline
					$R_f$   & 0.13$\%$  & Interest Paid & 18.95 & $A$   & 292.98    \\ \hline
					$R_m$   & 38.32$\%$ & $D_{2013}$    & 157.55    & $E$   & 135.43    \\ \hline
					$\beta$ & 1.43      & $D_{2012}$     & 117.05    & $w_E$ & 46.22$\%$ \\ \hline
					$C_E$   & 54.59$\%$ & $C_D$         & 13.80$\%$ & $w_D$ & 53.78$\%$ \\ \hline
				\end{tabular}
		\end{minipage}}
	\end{table}
	
	\begin{table}[H]
		\caption{Optimal values of $\alpha$ in \eqref{eq:objective_positive} and \eqref{eq:objective_negative} for different $q$ and $\mu$ with the end of December 2013 (see Tables \ref{tbl:american-apparel} and \ref{tbl:american-apparel-2}) as a reference point. $\Gamma=0.4$, $y=0$, $c=-2.0862$, $t=1$. The initial value of $\alpha$ for the optimization is set to $-1$. The values are displayed up to 4 decimal points.}
		\label{my-label}
		\resizebox{0.7\textwidth}{!}{\begin{minipage}{\textwidth}
				\hspace{-2cm}\begin{tabular}{c|cccccccccccccc}
					$\mu\backslash q$ & 0.01   & 0.04    & 0.07    & 0.1     & 0.13    & 0.16    & 0.19    & 0.22    & 0.25    & 0.28    & 0.31    & 0.34    & 0.37    & 0.4     \\\hline
					0.5               & 0.0000 & -2.0777 & -2.0811 & -2.0825 & -2.0828 & -2.0833 & -2.0835 & -2.0839 & -2.0841 & -2.0842 & -2.0843 & -2.0844 & -2.0845 & -2.0846 \\
					0.4               & 0.0000 & 0.0000  & -2.0798 & -2.0813 & -2.0819 & -2.0828 & -2.0831 & -2.0831 & -2.0834 & -2.0835 & -2.0837 & -2.0838 & -2.0839 & -2.0840 \\
					0.3               & 0.0000 & 0.0000  & -2.0777 & -2.0800 & -2.0808 & -2.0815 & -2.0821 & -2.0823 & -2.0826 & -2.0831 & -2.0832 & -2.0834 & -2.0835 & -2.0837 \\
					0.2               & 0.0000 & 0.0000  & 0.0000  & -2.0777 & -2.0792 & -2.0806 & -2.0812 & -2.0814 & -2.0817 & -2.0820 & -2.0823 & -2.0826 & -2.0827 & -2.0828 \\
					0.1               & 0.0000 & 0.0000  & 0.0000  & 0.0000  & -2.0778 & -2.0791 & -2.0802 & -2.0805 & -2.0809 & -2.0813 & -2.0816 & -2.0819 & -2.0821 & -2.0822 \\\hline
					-0.1              & 0.0000 & -2.0788 & -2.0820 & -2.0830 & -2.0835 & -2.0838 & -2.0842 & -2.0844 & -2.0845 & -2.0846 & -2.0846 & -2.0847 & -2.0847 & -2.0848 \\
					-0.15             & 0.0000 & -2.0792 & -2.0819 & -2.0829 & -2.0834 & -2.0837 & -2.0842 & -2.0843 & -2.0845 & -2.0846 & -2.0847 & -2.0847 & -2.0847 & -2.0848 \\
					-0.2              & 0.0000 & 0.0000  & -2.0819 & -2.0831 & -2.0834 & -2.0837 & -2.0841 & -2.0843 & -2.0844 & -2.0846 & -2.0847 & -2.0847 & -2.0848 & -2.0848 \\
					-0.25             & 0.0000 & 0.0000  & -2.0816 & -2.0830 & -2.0834 & -2.0837 & -2.0841 & -2.0843 & -2.0844 & -2.0845 & -2.0846 & -2.0847 & -2.0848 & -2.0848 \\
					-0.3              & 0.0000 & 0.0000  & -2.0816 & -2.0826 & -2.0835 & -2.0837 & -2.0839 & -2.0843 & -2.0844 & -2.0845 & -2.0846 & -2.0847 & -2.0848 & -2.0849 \\
					-0.35             & 0.0000 & 0.0000  & -2.0821 & -2.0826 & -2.0834 & -2.0838 & -2.0839 & -2.0843 & -2.0844 & -2.0845 & -2.0846 & -2.0847 & -2.0848 & -2.0848 \\
					-0.4              & 0.0000 & 0.0000  & 0.0000  & -2.0827 & -2.0834 & -2.0838 & -2.0839 & -2.0841 & -2.0844 & -2.0845 & -2.0846 & -2.0847 & -2.0846 & -2.0847 \\
					-0.45             & 0.0000 & 0.0000  & 0.0000  & -2.0828 & -2.0832 & -2.0836 & -2.0840 & -2.0841 & -2.0842 & -2.0845 & -2.0846 & -2.0847 & -2.0848 & -2.0848 \\
					-0.5              & 0.0000 & 0.0000  & 0.0000  & 0.0000  & -2.0833 & -2.0837 & -2.0840 & -2.0842 & -2.0842 & -2.0846 & -2.0846 & -2.0847 & -2.0847 & -2.0848 \\
					-0.55             & 0.0000 & 0.0000  & 0.0000  & 0.0000  & -2.0832 & -2.0837 & -2.0840 & -2.0842 & -2.0843 & -2.0846 & -2.0847 & -2.0847 & -2.0847 & -2.0848 \\
					-0.6              & 0.0000 & 0.0000  & 0.0000  & 0.0000  & -2.0849 & -2.0836 & -2.0839 & -2.0843 & -2.0843 & -2.0845 & -2.0847 & -2.0847 & -2.0847 & -2.0848 \\
					-0.65             & 0.0000 & 0.0000  & 0.0000  & 0.0000  & 0.0000  & -2.0835 & -2.0840 & -2.0841 & -2.0844 & -2.0844 & -2.0846 & -2.0847 & -2.0848 & -2.0848 \\
					-0.7              & 0.0000 & 0.0000  & 0.0000  & 0.0000  & 0.0000  & 0.0000  & -2.0838 & -2.0842 & -2.0843 & -2.0845 & -2.0845 & -2.0846 & -2.0848 & -2.0849 \\
					-0.75             & 0.0000 & 0.0000  & 0.0000  & 0.0000  & 0.0000  & 0.0000  & -2.0839 & -2.0840 & -2.0843 & -2.0844 & -2.0846 & -2.0846 & -2.0847 & -2.0848 \\
					-0.8              & 0.0000 & 0.0000  & 0.0000  & 0.0000  & 0.0000  & 0.0000  & 0.0000  & -2.0846 & -2.0847 & -2.0845 & -2.0845 & -2.0847 & -2.0847 & -2.0848 \\
					-0.85             & 0.0000 & 0.0000  & 0.0000  & 0.0000  & 0.0000  & 0.0000  & 0.0000  & -2.0845 & -2.0847 & -2.0848 & -2.0846 & -2.0846 & -2.0848 & -2.0848 \\
					-0.9              & 0.0000 & 0.0000  & 0.0000  & 0.0000  & 0.0000  & 0.0000  & 0.0000  & 0.0000  & -2.0843 & -2.0848 & -2.0849 & -2.0846 & -2.0847 & -2.0848 \\
					-0.95             & 0.0000 & 0.0000  & 0.0000  & 0.0000  & 0.0000  & 0.0000  & 0.0000  & 0.0000  & 0.0000  & -2.0844 & -2.0849 & -2.0850 & -2.0847 & -2.0851 \\
					-1                & 0.0000 & 0.0000  & 0.0000  & 0.0000  & 0.0000  & 0.0000  & 0.0000  & 0.0000  & 0.0000  & -2.0844 & -2.0845 & -2.0846 & -2.0850 & -2.0848 \\
					-1.05             & 0.0000 & 0.0000  & 0.0000  & 0.0000  & 0.0000  & 0.0000  & 0.0000  & 0.0000  & 0.0000  & 0.0000  & -2.0845 & -2.0846 & -2.0847 & -2.0851 \\
					-1.1              & 0.0000 & 0.0000  & 0.0000  & 0.0000  & 0.0000  & 0.0000  & 0.0000  & 0.0000  & 0.0000  & 0.0000  & 0.0000  & -2.0846 & -2.0847 & -2.0848 \\
					-1.15             & 0.0000 & 0.0000  & 0.0000  & 0.0000  & 0.0000  & 0.0000  & 0.0000  & 0.0000  & 0.0000  & 0.0000  & 0.0000  & -2.0855 & -2.0847 & -2.0848 \\
					-1.2              & 0.0000 & 0.0000  & 0.0000  & 0.0000  & 0.0000  & 0.0000  & 0.0000  & 0.0000  & 0.0000  & 0.0000  & 0.0000  & -0.0001 & -2.0850 & -2.0847 \\
					-1.25             & 0.0000 & 0.0000  & 0.0000  & 0.0000  & 0.0000  & 0.0000  & 0.0000  & 0.0000  & 0.0000  & 0.0000  & 0.0000  & -0.0001 & -0.0003 & -2.0847 \\
					-1.3              & 0.0000 & 0.0000  & 0.0000  & 0.0000  & 0.0000  & 0.0000  & 0.0000  & 0.0000  & 0.0000  & 0.0000  & 0.0000  & -0.0001 & -0.0224 & -2.0847 \\
					-1.35             & 0.0000 & 0.0000  & 0.0000  & 0.0000  & 0.0000  & 0.0000  & 0.0000  & 0.0000  & 0.0000  & 0.0000  & -0.0001 & -0.0019 & -0.0723 & -0.1551 \\
					-1.4              & 0.0000 & 0.0000  & 0.0000  & 0.0000  & 0.0000  & 0.0000  & 0.0000  & 0.0000  & 0.0000  & 0.0000  & -0.0002 & -0.0434 & -0.1216 & -0.2049 \\
					-1.45             & 0.0000 & 0.0000  & 0.0000  & 0.0000  & 0.0000  & 0.0000  & 0.0000  & 0.0000  & 0.0000  & -0.0001 & -0.0151 & -0.0916 & -0.1703 & -0.2542 \\
					-1.5              & 0.0000 & 0.0000  & 0.0000  & 0.0000  & 0.0000  & 0.0000  & 0.0000  & 0.0000  & -0.0001 & -0.0004 & -0.0619 & -0.1394 & -0.2186 & -0.3029 \\
					-1.55             & 0.0000 & 0.0000  & 0.0000  & 0.0000  & 0.0000  & 0.0000  & 0.0000  & 0.0000  & -0.0001 & -0.0304 & -0.1090 & -0.1870 & -0.2666 & -0.3512 \\
					-1.6              & 0.0000 & 0.0000  & 0.0000  & 0.0000  & 0.0000  & 0.0000  & 0.0000  & 0.0000  & -0.0019 & -0.0771 & -0.1560 & -0.2343 & -0.3143 & -0.3993 \\
					-1.65             & 0.0000 & 0.0000  & 0.0000  & 0.0000  & 0.0000  & 0.0000  & 0.0000  & -0.0001 & -0.0419 & -0.1236 & -0.2029 & -0.2815 & -0.3618 & -0.4472 \\
					-1.7              & 0.0000 & 0.0000  & 0.0000  & 0.0000  & 0.0000  & 0.0000  & -0.0001 & -0.0043 & -0.0882 & -0.1701 & -0.2496 & -0.3285 & -0.4092 & -0.4952 \\
					-1.75             & 0.0000 & 0.0000  & 0.0000  & 0.0000  & 0.0000  & 0.0000  & -0.0001 & -0.0479 & -0.1344 & -0.2166 & -0.2962 & -0.3755 & -0.4566 & -0.5432 \\
					-1.8              & 0.0000 & 0.0000  & 0.0000  & 0.0000  & 0.0000  & -0.0001 & -0.0038 & -0.0940 & -0.1806 & -0.2629 & -0.3428 & -0.4224 & -0.5041 & -0.5914 \\
					-1.85             & 0.0000 & 0.0000  & 0.0000  & 0.0000  & 0.0000  & -0.0001 & -0.0467 & -0.1401 & -0.2268 & -0.3093 & -0.3895 & -0.4694 & -0.5516 & -0.6399 \\
					-1.9              & 0.0000 & 0.0000  & 0.0000  & 0.0000  & 0.0000  & -0.0010 & -0.0928 & -0.1862 & -0.2730 & -0.3556 & -0.4361 & -0.5165 & -0.5994 & -0.6888 \\
					-1.95             & 0.0000 & 0.0000  & 0.0000  & 0.0000  & -0.0001 & -0.0359 & -0.1389 & -0.2323 & -0.3192 & -0.4020 & -0.4829 & -0.5638 & -0.6475 & -0.7383 \\
					-2                & 0.0000 & 0.0000  & 0.0000  & 0.0000  & -0.0002 & -0.0820 & -0.1850 & -0.2784 & -0.3655 & -0.4485 & -0.5297 & -0.6113 & -0.6960 & -0.7886
				\end{tabular}
		\end{minipage}}\label{tbl:comparative_alpha}
	\end{table}
\end{appendix}

%\bibliographystyle{plainnat}
%\bibliography{references}
\end{document}